\documentclass[12pt]{article}
\usepackage[margin=.9in]{geometry}
\usepackage{XCharter}
\usepackage[english]{babel}
\usepackage{ae,aecompl,amsmath,amsbsy,amssymb,eurosym,amsthm,bm}
\usepackage{mathtools}
\usepackage{booktabs}
\usepackage{amsfonts}
\usepackage{graphicx}
\usepackage{tikz}
\usepackage{caption}
\usepackage{makecell}

\usepackage{newunicodechar}
\newunicodechar{̇}{}

\usepackage{comment}
\usepackage{color}
\usepackage{pdflscape}
\usepackage{enumerate}
\usepackage{rotating}
\usepackage{adjustbox,algorithm2e}

\usepackage{setspace}
\usepackage{url}
\usepackage{xr} 
\usepackage{xr-hyper}
\usepackage[colorlinks]{hyperref}
\hypersetup{
  colorlinks=true,
  linkcolor=blue,
  citecolor=blue,      
  urlcolor=cyan,
  filecolor=orange 
}

\usepackage{fancyhdr}
\usepackage{bbm}
\usepackage[multiple]{footmisc}

\usepackage[round]{natbib}

\usepackage[]{caption}
\usepackage{graphicx}
\usepackage{longtable}
\usepackage{float}
\usepackage{tikz}
\usepackage{multirow}
\usepackage[flushleft]{threeparttable}
\usepackage{subcaption}
\usepackage{booktabs}
\usepackage{tabularx}
\newcolumntype{s}{>{\hsize=.33\hsize}X}

\usepackage{tabularray}
\usepackage{float}
\usepackage{graphicx}
\usepackage{codehigh}
\usepackage[normalem]{ulem}
\UseTblrLibrary{booktabs}
\UseTblrLibrary{siunitx}

\NewTableCommand{\tinytableDefineColor}[3]{\definecolor{
#1}{#2}{#3}}

\usepackage{tikz}
\usetikzlibrary{matrix, positioning,patterns}

\makeatletter
\newcommand*{\addFileDependency}[1]{
  \typeout{(#1)}
  \@addtofilelist{#1}
  \IfFileExists{#1}{}{\typeout{No file #1.}}
}
\makeatother

\newtheorem{assumption}{Assumption}[section]
\newtheorem{theorem}{Theorem}[section]
\newtheorem{lemma}{Lemma}[section]
\newtheorem{prop}{Proposition}[section]
\newtheorem{corollary}{Corollary}[section]
\newtheorem{definition}{Definition}
\theoremstyle{definition}

\newtheorem{remark}{Remark}[section]
\newtheorem{proposition}{Proposition}[section]

\renewcommand{\Pr}{\mathbb{P}}
\newcommand{\R}{\mathbb{R}}
\newcommand{\1}{\mathbf{1}}
\newcommand{\E}{\mathbb{E}}

\newcommand{\ind}{\mathbf{1}}

\title{Evaluating Local Policies in Centralized Markets\thanks{{\small We thank Kevin Chen, Ramesh Johari, Lihua Lei, Evan Munro, Davide Viviano, and Stefan Wager for their feedback, which has greatly improved the paper. We also thank seminar participants at Harvard and Stanford, as well as conference participants at the 2024 Winter Meetings of the Econometric Society, and ACIC 2025. Wisse Rutgers acknowledges financial support from the Fundación Ramón Areces, the Maria de Maeztu Unit of Excellence CEX2020-001104-M, funded by MCIN/AEI/10.13039/501100011033, and CEMFI. }}} 

\author{Dmitry  Arkhangelsky \thanks{{\small  Associate Professor, CEMFI, darkhangel@cemfi.es. }} \and Wisse Rutgers \thanks{{\small PhD student, CEMFI, wisse.rutgers@cemfi.edu.es.}}}

\date{\today}

\begin{document}
\onehalfspacing

\begin{titlepage}
    \maketitle
\begin{abstract}
We study a policy evaluation problem in centralized markets. We show that the aggregate impact of any marginal reform, the Marginal Policy Effect (MPE), is nonparametrically identified using data from a baseline equilibrium, without additional variation in the policy rule. We achieve this by constructing the equilibrium-adjusted outcome: a policy-invariant structural object that augments an agent's outcome with the full equilibrium externality their participation imposes on others. We show that these externalities can be constructed using estimands that are already common in empirical work. The MPE is identified as the covariance between our structural outcome and the reform's direction, providing a flexible tool for optimal policy targeting and a novel bridge to the Marginal Treatment Effects literature.
\end{abstract}
    
    \vspace{2\baselineskip}
    
    
    \thispagestyle{empty}
\end{titlepage}

\section{Introduction}

Centralized marketplaces are a cornerstone of the modern economy, organizing a vast and growing share of economic activity. In the digital world, they match riders with drivers, allocate advertising slots in real-time auctions, and connect millions of sellers to buyers. In the public and non-profit sectors, they assign students to schools, allocate housing vouchers, and match organ donors to recipients. The defining feature of these markets is a well-defined algorithm or set of rules that processes inputs from participants, e.g., bids, preferences, scores, and produces an allocation of scarce resources.

A fundamental question for both the designers and regulators of these marketplaces is how to improve the outcomes they generate. This paper focuses on a particularly common class of interventions: those that influence participants' behavior within a fixed set of market rules. This includes policies like providing subsidies to certain users, offering informational nudges, or creating incentives to alter how they participate in the market. Given an existing policy instrument, the central challenge for a platform or regulator is how to optimize it. This optimization is often an iterative process, focused on the aggregate welfare consequences of a marginal adjustment, for example, slightly expanding eligibility for a fee waiver or tweaking the size of a subsidy.

The natural approach to evaluate such a change is to run an experiment. However, standard experimentation in these environments faces a well-known challenge: equilibrium spillovers. Any intervention that meaningfully changes the behavior of one group of participants induces an endogenous response from the market-clearing mechanism that affects all other participants. For example, a subsidy that encourages more applications to a university with fixed capacity will raise the admission cutoff, creating a negative spillover for all other applicants. Because this spillover affects both treated and control groups alike, a simple comparison between them would difference away this common, system-wide component. 

This issue, sometimes called the "missing intercept problem" in macroeconomics \citep{wolf2023missing}, has led to a conventional wisdom that nonparametric identification of aggregate effects requires observing the system's response to explicit variation in the policy environment itself. Researchers typically seek this variation either over time, as is common in industry with switchback experiments \citep{bojinov2023design}, or across distinct economic contexts where the policy rule or its intensity differs. For instance, a common design in development economics involves a two-stage randomization where the share of participants receiving a benefit is experimentally varied across different local labor markets (e.g., \citealp{crepon2013labor}). However, such designs can be costly to implement, their findings may be difficult to interpret due to substantial heterogeneity across environments, and finally, the statistical results might lack power due to a small number of experimental units.

Our central contribution is to show that aggregate effects of policy changes are, in fact, nonparametrically identified from data within a single policy environment. We achieve this without requiring cross-market variation or imposing strong extrapolating assumptions. The key to our approach is the construction of a single, policy-invariant structural object for each agent: the equilibrium-adjusted outcome ($\Psi_{i}^{\text{total}}$). This object augments an agent's observed outcome with a correction term that captures the full equilibrium externality their participation imposes on all others. The result is a single measure of an agent's total contribution to welfare, accounting for both their private outcomes and the full cost of the competitive pressure they exert on the system.  We show that this key theoretical object is identified from the data under natural assumptions about the market structure. 

Our framework uses this structural object to evaluate the effects of local reforms---that is, marginal adjustments to an existing policy, such as slightly increasing the share of participants who receive a subsidy. Any such reform is characterized by a "score" function, $s_W$, which describes the precise direction of the change. For instance, a reform that marginally increases the share of treated individuals would be represented by a score that is positive for the treated group and negative for the untreated group, capturing the small shift of participants between them. This construction leads to a powerful "separation principle" that is the main practical result of our paper. We show that the Marginal Policy Effect (MPE)---the first-order welfare impact of the reform---can be expressed as a simple covariance between this score and our structural outcome:
$$
\text{MPE} = \mathbb{E}[\Psi_{i}^{\text{total}} \cdot s_{W}(W_{i})].
$$
This result provides a practical tool for policy evaluation. It separates the complex, fixed market structure, which is entirely encapsulated in $\Psi_{i}^{\text{total}}$, from the specific policy change under consideration, which is represented by the score $s_W$. A researcher or platform can invest in estimating the structural object $\Psi_{i}^{\text{total}}$ once and then use it to evaluate any local reform. 

While we frame our discussion in terms of average outcomes for clarity, the framework is substantially more general. It applies to any welfare criterion that is a smooth functional of the outcome distribution, allowing policymakers to evaluate a reform's impact on quantiles, measures of inequality like the Gini coefficient, or other distributional objectives. This generality is a direct consequence of our focus on local reforms. For a marginal policy change, the first-order impact on any smooth distributional functional can be represented as an expectation of a specific, policy-invariant transformation of the outcome, known as its influence function. Our analysis, therefore, proceeds by first developing the results for the simple case where welfare is the average outcome, and we later demonstrate that our framework accommodates these more general criteria by simply substituting the outcome variable with its relevant influence function.

The identification of different components of $\Psi_i^{\text{total}}$ is not straightforward. The key technical hurdle is that centralized allocation mechanisms are often discontinuous. A school admission rule, for example, is a step function of a student's test score. A marginal policy reform that infinitesimally raises an admission cutoff has no effect on most students, but it causes a discrete jump in the allocation of students right at the margin, who now lose their seats. This creates a fundamental identification challenge, as the observed data contains no direct information about the outcomes of these marginal students under their new, counterfactual allocation.

Our framework resolves this by formally characterizing the indirect, equilibrium component of a policy's effect and showing that it can be identified by focusing on the "marginal agents" at the allocation boundary. This approach builds a direct bridge between the theory of market design and a large body of empirical work. We show that the crucial inputs required to compute the market externalities are often the same local average treatment effects (LATEs) identified in regression discontinuity (RDD) studies of admission cutoffs or randomized lotteries (e.g., \citealp{abdulkadirouglu2017research,abdulkadi̇rouglu2022breaking,kirkeboen2016field,walters2018demand}). Our framework clarifies that these well-studied parameters are not merely reduced-form objects but essential structural inputs required to conduct a full equilibrium evaluation of any marginal policy change (see \citet{kline2016evaluating} for a related discussion).

The flexibility of our framework allows us to extend the analysis beyond idealized experiments to more common observational settings. We first consider the case of selection on observables, where a policy is assigned randomly conditional on a set of covariates. The primary application of this extension is to provide a rigorous foundation for optimal policy targeting, connecting our results to the literature on empirical welfare maximization (EWM) \citep{manski2004statistical, kitagawa2018should, athey2021policy,viviano2024policy}. A practical implication of our framework is that it can avoid a common curse of dimensionality. As long as the market mechanism is anonymous with respect to the covariates---reacting only to agents' reports, not their background characteristics---the structural components of $\Psi_{i}^{\text{total}}$ do not need to be re-estimated conditionally. A researcher can proceed directly with an aggregate analysis, using unconditionally estimated parameters like the RDD effects.

Finally, the framework can be adapted to answer a different and more structural class of questions central to economic analysis. In many contexts, a policymaker cannot directly mandate an action because it is an agent's endogenous choice, such as the decision to apply for a voucher or take up a program. We show that our framework can still be used to evaluate the welfare consequences of a marginal shift in the distribution of these choices. This is achieved by introducing an instrumental variable that provides exogenous variation in agents' decisions, building a novel bridge to the literature on Marginal Treatment Effects (MTE) \citep{bjorklund1987estimation, heckman2001policy, heckman2005structural}. We show that the MTE of the equilibrium-adjusted outcome is precisely the correct structural object for evaluating policies that operate by influencing agents' choices. This connects the estimands from an IV analysis to the MPE for specific, economically meaningful reforms, providing a powerful tool for structural evaluation in the presence of endogenous selection.

Our analysis is local, providing the welfare gradient for marginal policy reforms rather than evaluating large-scale, global changes. This focus is necessitated by a fundamental identification challenge: a marginal change to a market-clearing cutoff can assign agents to allocations they never would have received in the baseline equilibrium, meaning their counterfactual outcomes are unobserved. While methods for identifying global effects exist \citep{munro2025causal}, they often rely on strong assumptions to bridge this identification gap. We show that even identifying the local welfare effect in the presence of these discontinuities is a non-trivial problem that requires a dedicated framework. 

To isolate this core challenge, our framework makes several simplifying assumptions. First, we focus exclusively on spillovers transmitted through the market-clearing mechanism itself. We abstract from other empirically important channels of interference, such as peer effects where agents' preferences respond to the allocation of others \citep{allende2019competition,leshno2022stable}, strategic reporting in non-strategy-proof environments \citep{agarwal2018demand,bertanha2024causal}, and exogenous policy spillovers like information diffusion. 
Second, we focus on identification rather than estimation or inference. This involves assuming that the researcher observes all relevant data, including agents' full reports, thereby abstracting from important practical challenges such as mistakes or incomplete preference rankings that are the subject of a separate literature \citep{artemov2023stable, fack2019beyond}. 

We adopt these limitations not because these other channels are unimportant, but to establish what is possible in an ideal setting. Our finding that identification is limited to local effects even under these demanding assumptions suggests that incorporating additional complexities would require further, potentially less credible, restrictions. In particular, after establishing our key results, we discuss the issue of strategic reporting, arguing that existing solutions \citep{agarwal2018demand, bertanha2023causal} can only partially address the challenges that arise in our context.

In summary, our framework provides a bridge between reduced-form data and the equilibrium structure of the market, yielding a tool with several distinct applications. First, our results can be applied directly to policy optimization. While our analysis is local, many practical policy decisions are iterative and marginal in nature, such as tweaking the size of a subsidy or adjusting eligibility criteria. Our framework provides the precise tool needed to guide these decisions by evaluating the "bang-for-the-buck" of a wide range of potential local reforms. Second, our work serves as a disciplined first step toward evaluating global policy changes. By sharply delineating what is identified from the data, it provides a transparent foundation upon which any analysis of large-scale reforms must be built; any claim about global effects must necessarily rely on extrapolating from the local effects that we identify. Finally, and relatedly, our results can be used to assess and discipline more ambitious structural models. A global model, likely estimated under more restrictive conditions, should be able to reproduce the local equilibrium effects that our framework identifies from the data, providing a powerful, data-driven specification test for more complex models of market behavior.

Our work is situated at the intersection of several active research areas: causal inference in markets, sufficient statistics approach in public economics, the empirical analysis of market design, the literature on optimal policy targeting, and the analysis of treatment effects with endogenous selection.

Our paper contributes to the recent literature on causal inference in the presence of interference and market equilibrium effects. The challenge that equilibrium adjustments can invalidate standard treatment effect comparisons has long been recognized \citep{heckman1998general}. One prominent branch of the recent literature leverages auxiliary experimental variation---for instance, randomized prices---to identify spillovers \citep{wager2021experimenting, munro2025treatment}. Another branch, closer to our own, relies on institutional knowledge of the market-clearing rule \citep{munro2025causal}. Our framework follows this latter approach but makes a distinct contribution by focusing on general, stochastic downstream outcomes (e.g., future earnings) rather than on the allocation itself or its deterministic functions. This broader scope for the outcome variable is what necessitates our focus on local, rather than global, policy effects. Our work also provides a clear economic structure for the statistical decompositions of interference proposed in the causal inference literature \citep{hu2022average}, showing precisely how the indirect effects arise from the market mechanism. Finally, it is related to the recent design-based causal analysis of equilibrium systems by \citet{menzel2025fixed}.

Our work is related to the influential sufficient statistics literature in public economics, which connects credibly identified, reduced-form parameters to welfare theory without requiring the estimation of a full structural model (e.g., \citealp{chetty2009sufficient,kleven2021sufficient}). In particular, in constructing the equilibrium-adjusted outcome, we explicitly rely on quasi-experimental estimands (such as LATEs from RDDs), combining them with institutional knowledge of the market mechanism. Our approach is more structural in nature, relying on details of the allocation mechanism.

Our analysis is directly related to applied market design. A prominent empirical literature uses randomized lotteries or regression discontinuity designs to estimate the causal effect of attending a particular school \citep{abdulkadirouglu2017research, abdulkadi̇rouglu2022breaking,walters2018demand}. We show that the LATEs identified in these studies are not merely reduced-form parameters. Instead, they are the essential structural inputs required to evaluate the equilibrium consequences of any marginal policy change. Our analysis builds on theoretical work that characterizes large matching markets with cutoffs \citep{azevedo2016supply, leshno2021cutoff} and also speaks to the econometric challenges that arise in non-strategy-proof mechanisms, as studied in a growing literature on preference recovery and strategic reporting \citep{agarwal2018demand, bertanha2023causal}. Finally, our analysis is related to recent empirical work on outcomes and choices in empirical market design by \citep{agarwal2025choices}

By focusing on policy optimization, our work connects to the literature on optimal policy targeting \citep{manski2004statistical, kitagawa2018should, athey2021policy,viviano2024policy}. This literature typically seeks to find a globally optimal policy rule, which often depends only on the sign of a conditional average treatment effect (CATE). Our approach is local, focusing on the welfare gradient to guide iterative policy improvement. Our central contribution to this literature is to identify the correct welfare-relevant object for policy targeting in an equilibrium setting. We show that the policymaker's objective should be to maximize the CATE of the equilibrium-adjusted outcome, $\Psi_{i}^{\text{total}}$, not the observed outcome. This objective function correctly accounts for equilibrium spillovers and leverages the magnitude of the causal effect, not just its sign.

Finally, to address settings with endogenous policy take-up, we connect to the literature on MTE \citep{bjorklund1987estimation, heckman2001policy,heckman2005structural}. In contexts where a policy instrument influences, but does not mandate, an agent's choice, we show how to use instrumental variables to conduct a full welfare analysis. Our key contribution is to demonstrate that the proper object of study is the MTE of the equilibrium-adjusted outcome. This synthesizes the MTE framework, which accounts for selection on unobservables, with our framework, which accounts for equilibrium spillovers. By connecting our equilibrium analysis to recent advances in the MTE literature \citep{brinch2017beyond,mogstad2018using,mogstad2024instrumental}, this result provides a clear path from reduced-form IV estimates to a rich set of structural statements about the welfare impact of policies that target endogenous choices.

The remainder of the paper is organized as follows. Section~\ref{sec:framework} lays out the theoretical framework, defining the economic environment and the propagation of a policy reform. Section~\ref{sec:derivation} presents our main identification result, detailing the construction of the equilibrium-adjusted outcome and the derivation of the Marginal Policy Effect. Section~\ref{sec:examples} illustrates the framework with several canonical examples, including auctions and school choice. Section~\ref{sec:extensions} develops the extensions to general welfare functionals, optimal policy targeting, and endogenous selection. Section~\ref{sec:conclusion} concludes.
\section{Framework}\label{sec:framework}

This section develops the formal framework to address the challenge of equilibrium spillovers outlined in the introduction. Standard causal inference methods are ill-suited for this environment, as a policy change for one agent directly alters the competitive landscape for all others. Our analysis overcomes this "missing intercept problem" by leveraging institutional knowledge of the market. To do so, we first specify the economic environment, defining the two components of institutional context on which our analysis rests: the allocation mechanism and the market conduct rule. We then trace how these components allow a local policy reform to propagate through the system.

\subsection{Environment}

Our analysis begins with a population of agents, indexed by $i$. Each agent has a potential outcome, $Y_i(w, a)$, which is a function of two variables: the allocation they ultimately receive, $a$, and their exposure to a policy instrument, $w$. The allocation $a$ belongs to a discrete set $\mathcal{A} = \{0, 1, \dots, K\}$, representing one of $K$ scarce goods---such as a seat at a charter school, a housing voucher, or a specific advertising slot---or an outside option ($a=0$). The policy instrument (or "treatment") $w \in \mathcal{W}$ represents an existing intervention that a platform or regulator is considering adjusting. Examples include the size of a tuition subsidy, an informational treatment about market options, or a targeted incentive. We denote the realized policy for agent $i$ by the random variable $W_i$.

In addition to the outcome, the policy $W_i$ influences agent $i$'s report to the allocation mechanism, $R_i = R_i(W_i)\in \mathcal{R}$. This report can correspond to a vector of preferences, a bid, or school priorities. An agent's final allocation, $A_i$, is determined by their own report $R_i$ and the aggregate competitive environment. We summarize this environment with a vector of market-clearing parameters, $\mathbf{c}$ (e.g., equilibrium prices, rationing probabilities, or admission cutoffs), and the population-wide distribution of reports, $P_R$.
\begin{assumption}[Anonymous Allocation Mechanism] \label{ass:mechanism}
The counterfactual allocation $A_i(r, \mathbf{c}, P_R)$ is determined by an anonymous mechanism that depends on an agent's report $r$, the common parameter $\mathbf{c}$, and the counterfactual marginal distribution of reports in the population $P_R$. The probability of receiving allocation $a$ is given by a \textbf{known} function $\mu_a(r, \mathbf{c}, P_R)$.
\end{assumption}
This framework is designed to capture two distinct but related economic settings. The first is a large market where an agent's allocation depends on their report relative to aggregate competitive conditions. In many such markets, these conditions are fully summarized by the equilibrium parameter $\mathbf{c}$, rendering the direct dependence of $\mu_a$ on $P_R$ redundant once $\mathbf{c}$ is known. The second setting is a market with a finite number of symmetric participants, such as a symmetric auction. Here, $\mu_a(r, \mathbf{c}, P_R)$ represents an agent's interim probability of receiving allocation $a$, which naturally depends on both the common parameter (e.g., a reserve price) and the distribution of their opponents' reports, $P_R$. Our general formulation, $\mu_a(r, \mathbf{c}, P_R)$, is deliberately chosen to encompass both of these cases.

Crucially, we assume the functional form of this allocation rule is known to the researcher. This institutional knowledge is essential for analyzing counterfactual allocations under different market conditions. Together, these two assumptions formalize the institutional knowledge of the market's structure. This knowledge is an essential component of our identification strategy, allowing us to proceed without requiring the cross-policy variation used in conventional approaches. Our next assumption extends this requirement from the mechanism's rules to the market's conduct.
\begin{assumption}[Market Conduct Rule] \label{ass:equilibrium_adj}
The counterfactual parameter vector $\mathbf{c}$ is determined by the counterfactual competitive environment $P_R$,
\begin{equation*}
    \mathbf{c} = \mathbf{c}(P_R),
\end{equation*}
for some \textbf{known} function $\mathbf{c}(\cdot)$.
\end{assumption}

We deliberately separate the market-clearing parameter $\mathbf{c}$ from the full report distribution $P_R$ as arguments in the allocation function, $\mu_a(r, \mathbf{c}, P_R)$. This distinction is not merely notational but economically and mathematically meaningful. Economically, it reflects a natural hierarchy: $P_R$ represents the primitive competitive environment, while $\mathbf{c}$ is the endogenous summary statistic of that environment (e.g., a vector of prices or admission cutoffs) to which agents directly react. Mathematically, this separation allows us to impose different regularity conditions on each component. For instance, an agent's allocation is often a discontinuous function of the clearing parameter $\mathbf{c}$. In contrast, the market conduct rule $\mathbf{c}(P_R)$ can be a smooth functional of the underlying distribution $P_R$. This structure is critical for analyzing the propagation of marginal policy reforms, as the examples below will illustrate.

\paragraph{Example 1: Competitive Equilibrium.}
Consider a market where for each product $k \in \{1, \dots, K\}$, there is a fixed supply $q_k$. The market-clearing parameter $\mathbf{c} \in \mathbb{R}^{K}$ can be interpreted as a vector of prices, admission standards, or rationing probabilities that adjust to equilibrate demand with supply. An equilibrium in a typical allocation mechanism, e.g., the Deferred Acceptance Algorithm, will depend only on $\mathbf{c}$, not the whole distribution of reports $P_R$ \citep{azevedo2016supply}. The equilibrium constraint requires that these parameters satisfy market-clearing:
\begin{equation*}
    \E[\ind\{A_i = k\}] = \int \mu_k(r, \mathbf{c}) dP_{R}(r) = q_k, \quad \forall k \in \{1, \dots, K\}.
\end{equation*}
This system of equations implicitly defines the parameter $\mathbf{c}$ as a functional of the entire distribution of reports, $\mathbf{c}(P_R)$.

\paragraph{Example 2: Optimal Reserve Price in an Auction.}
Suppose a seller is auctioning a single item ($\mathcal{A} = \{0,1\}$) and agents' reports $R_i$ are their valuations. Parameter $c$ corresponds to a reserve price. In this situation, the allocation probability for an agent with valuation $r$ will depend discontinuously on the reserve price and smoothly on the distribution of the competitors' valuations. Suppose the seller sets a reserve price $c$ to maximize expected revenue. Following \citet{myerson1981optimal}, the optimal reserve price solves the first-order condition of the seller's problem, which balances the revenue gain from a higher price against the risk of the item going unsold. This trade-off is captured by the equation:
\begin{equation*}
    f_{R}(c)c - (1 - F_{R}(c)) = 0,
\end{equation*}
where $F_R$ and $f_R$ are the distribution and density functions of valuations, respectively. The solution, $c$, is the price where the marginal revenue of raising the price equals zero. Because this condition directly involves the distribution of valuations, the optimal reserve price is a function of the report distribution, $c(P_R)$. This example highlights the value of our two-part structure: Assumption~\ref{ass:mechanism} defines the allocation rule for any given reserve price, while the market conduct rule in Assumption~\ref{ass:equilibrium_adj} models the seller's optimizing behavior.

\paragraph{Example 3: School Choice with Trading Cycles.}
Suppose $K$ products correspond to spots in public schools. An agent's report consists of a vector of priorities, $(V_{i,1},\dots, V_{i,K})$, and a strict preference relation, $\succ_i$. The slots are allocated using Gale's Top Trading Cycles algorithm \citep{shapley1974cores,abdulkadirouglu2003school}. As shown by \citet{leshno2021cutoff}, the final allocation from this process can be characterized by a matrix of admission cutoffs $\mathbf{c} \in \mathbb{R}^{K\times K}$. Here, $c_{a,b}$ represents the minimum priority score required at an endowment school $b$ to successfully obtain a seat at a destination school $a$. A student is assigned their most-preferred school from the set of schools for which they are admissible:
\begin{equation*}
    A_i = \max_{\succ_i}\left\{\{ a |V_{i,b} \ge c_{a,b} \text{ for some $b$}\} \cup \{0\}\right\}.
\end{equation*}
These cutoffs are the endogenous outcome of the matching process, which can be described as a solution to a dynamic system that depends on the joint distribution of priorities and preferences, allowing us to write $\mathbf{c} = \mathbf{c}(P_R)$.\footnote{We discuss application of our results to TTC in more detail in Appendix~\ref{app:examples}.}
\vspace{2em}

The previous assumptions describe the counterfactual world. The observed data are generated from a baseline equilibrium that unfolds in a sequence of steps. First, the baseline policy distribution, $P_{W|0}$, induces a distribution of agent reports, $P_{R|0}$. This report distribution, via the market conduct rule in Assumption~\ref{ass:equilibrium_adj}, determines the parameter $\mathbf{c}_0 := \mathbf{c}(P_{R|0})$. Finally, the allocation mechanism from Assumption~\ref{ass:mechanism} assigns an allocation $A_i = A_i(R_i, \mathbf{c}_0, P_{R|0})$, which determines the realized outcome $Y_i := Y_i(W_i, A_i)$. We define the benchmark aggregate welfare, $\mathcal{U}_0$, as the expected value of this realized outcome:
\begin{align*}
    \mathcal{U}_0 := \E_{0}[Y_i],
\end{align*}
where the expectation $\mathbb{E}_{0}$ is taken over the distribution of the observed data described above. 

We will think of a marginal policy reform as a specific "direction" of change to the baseline policy distribution. Any such change can be characterized by a score function, $s_W(w)$, which tells us how the reform re-weights the baseline distribution of $W_i$. Intuitively, if $s_W(w)$ is positive, the reform slightly increases the proportion of individuals receiving policy $w$; if it is negative, it slightly decreases that proportion. A convenient way to generate a local reform is to embed the baseline policy distribution in a parametric family, $P_{W|\theta, s_{W}}$, indexed by a real-valued parameter $\theta$ and score $s_{W}$, so that $P_{W|0, s_{W}} = P_{W|0}$. A marginal reform is then represented by an infinitesimal change in $\theta$ away from its baseline value of zero. 

To make this more concrete, consider two examples. First, suppose the policy is a binary treatment ($W_i \in \{0, 1\}$) and the reform's goal is to marginally increase the share of treated agents. This corresponds to a score function that is proportional to $\frac{W_i}{\mathbb{E}_{0}[W_i]} - \frac{1-W_i}{1- \mathbb{E}_{0}[W_i]}$ and thus is positive for the treated ($s(1) > 0$) and negative for the untreated ($s(0) < 0$), effectively shifting a small amount of probability mass between the two groups. Alternatively, suppose the policy is a subsidy (in logs) distributed as $W_i \sim \mathcal{N}(a_0, \sigma^2_0)$, and the reform aims to reduce the subsidy's variance. This corresponds to a score function $s_W(w)$ that is high for subsidies near the mean and negative for subsidies in the tails, effectively pulling the distribution in towards its center.

Given $(\theta, s_{W})$ we can define the welfare as the average outcome, 
\begin{align*}
    \mathcal{U}(\theta, s_w) := \E_{(\theta, s_{W})}[Y_i],
\end{align*}
where the expectation is over the distribution of the outcomes induced by $(\theta,s_{W})$, which we will describe in detail in the next section. By definition $\mathcal{U}(0,s_{W}) = \mathcal{U}_0$. Following \citet{hu2022average}, our primary objective is to characterize the marginal policy effect (MPE), which measures the first-order impact of the reform on the aggregate outcome:
\begin{align*}
   \text{Marginal Policy Effect} :=  \frac{\partial \mathcal{U}(\theta,s_{W})}{\partial \theta}\bigg|_{\theta=0}.
\end{align*}
The next several sections of the paper will focus on the identification of the MPE for a fixed local reform (score $s_{W}$). We will discuss the question of optimal local reform in detail in Section~\ref{sec:extensions}.

Our analysis focuses on identification of the MPE, not estimation or inference. We therefore proceed as if the joint distribution of the data vector $D_i = (Y_i, A_i, W_i, R_i)$ under the baseline policy regime is known to the researcher. Note that it implies that the full report $R_i$ is observed. Part of this data is often missing in some empirical market design applications, such as auctions (e.g., observing the bids only for winners). The reports can contain mistakes or incomplete rankings in other applications, such as school choice \citep{artemov2023stable,fack2019beyond}. Our analysis abstracts from these data limitations.

\begin{remark}[Beyond average outcomes]\label{rem:in_local}
    Our discussion focuses on welfare, which is defined as an average outcome, but given that the reforms we analyze are local, our results apply almost immediately to any other smooth functionals of the outcome distribution. For instance, they can be extended to cover quantiles or aggregate measures of inequality, such as the Gini coefficient. We discuss this extension in Section~\ref{sec:extensions}.
\end{remark}

\subsection{From Policy to Likelihood: Tracing the Perturbation}

To trace how a local policy reform propagates through the market, we must connect the observed data to the counterfactual world. Our first step is to decompose the joint distribution of the observed data for agent $i$, $D_i = (Y_i, A_i, R_i, W_i)$, which we denote $P_{D}^{obs}$.  We assume this distribution has a density, $f_D^{obs}(y, a, r, w)$, with respect to an underlying well-behaved measure, which we factor as:
\begin{align*}
    f_{D}^{obs}(y, a, r, w) = f^{obs}_{Y|A,R,W}(y|a,r,w) f^{obs}_{A|R,W}(a|r,w) f^{obs}_{R|W}(r|w) f^{obs}_{W}(w).
\end{align*}
This factorization is purely statistical and holds by construction. Our goal in this section is to use this distribution to inform us about the counterfactual distribution of the data $P_{D}^{count}(\theta,s_{W})$ induced by a particular local reform.

Our model already imposes structure on this decomposition. Specifically, Assumptions~\ref{ass:mechanism}-\ref{ass:equilibrium_adj} imply that the observed conditional allocation probability, $f^{obs}_{A|R,W}(a|r,w)$ is given by the known allocation rule $\mu_{a}(\cdot)$ evaluated at the baseline equilibrium:
\begin{align*}
    f^{obs}_{A|R,W}(a|r,w) = \mu_a(r, \mathbf{c}_0, P_{R|0}).
\end{align*}
The counterfactual distribution of $W_i$ is controlled by the policy maker, and for a given $(\theta, s_{W})$, we denote its density by $f_{W}(w|\theta, s_{W})$. This leaves two components that we must understand: $f^{obs}_{Y|A,R,W}(y|a,r,w)$ and $f^{obs}_{R|W}(r|w)$. To do so, we start with an assumption about the assignment of the baseline policy.
\begin{assumption}[Random Assignment] \label{ass:random_assignment}
The policy $W_i$ is randomly assigned to agents in the baseline environment.
\end{assumption}
This assumption isolates the mechanism's spillover effects from confounding selection effects, simplifying our initial analysis. We relax this restriction in Section~\ref{sec:extensions} to allow for selection on both observed and unobserved characteristics. Our next assumption isolates the spillovers created by the allocation mechanism from other potential interference channels.
\begin{assumption}[Policy Invariance] \label{ass:structural_invariance}
The potential outcomes $Y_i(w,a)$ and potential reports $R_i(w)$ are structural primitives that are invariant to $(\theta,s_{W})$.
\end{assumption}
Assumption~\ref{ass:structural_invariance} requires that an agent's underlying potential outcomes and potential reports, $Y_{i}(w,a)$ and $R_i(w)$, do not respond to changes in the aggregate policy environment $(\theta, s_{W})$. This allows us to focus squarely on externalities transmitted through the allocation mechanism. The assumption would be violated in two main scenarios. First, if the mechanism were not strategy-proof, an agent's optimal report $R_i$ would depend on the distribution of competitors' reports, which is a function of $(\theta,s_{W})$. Second, if direct peer effects were present (e.g., an agent's utility is affected by the allocation of its competitors), both potential outcomes and potential reports would depend on the policy distribution (e.g., \citealp{allende2019competition,leshno2022stable}). By ruling these out at the outset, our framework provides a clean benchmark for understanding spillovers induced by the mechanism only. We discuss the strategic reporting channel in Section~\ref{sec:extensions}.

With these assumptions in place, we can now trace how the policy perturbation reshapes the joint distribution of the data.
\begin{prop}[Propagation of a Policy Perturbation]\label{prop:id_sup}
    Suppose Assumptions \ref{ass:mechanism}-\ref{ass:structural_invariance} hold. Then $P_{D}^{obs}$-almost surely we have
    \begin{align*}
    f_{D}^{\text{count}}(Y_i, A_i, R_i, W_i|\theta, s_{W}) = \; & f^{\text{obs}}_{Y|A,R,W}(Y_i|A_i,R_i,W_i) \mu_a(R_i, \mathbf{c}(P_{R|\theta,s_{W}}), P_{R|\theta,s_{W}}) \\
    & \times f^{\text{obs}}_{R|W}(R_i|W_i) f_{W}(W_i|\theta,s_{W})
\end{align*}
    where $f_{D}^{count}(y,a,r,w|\theta,s_{W})$ is the density of $P_{D}^{count}(\theta, s_{W})$, and  $P_{R|\theta,s_{W}}$ is the counterfactual distribution of reports induced by the new policy. Its density is formed by integrating the observed reporting rule, $f^{obs}_{R|W}(r|w)$, against the new policy distribution, $f(w|\theta, s_{W})$.
\end{prop}
Proposition~\ref{prop:id_sup} provides a crucial link between the observed data and the counterfactual world. The expression reveals that a policy reform propagates through two distinct channels: (1) a direct effect on agents from the change in the policy distribution itself, from $f_{W}^{obs}(w)$ to $f_W(w|\theta,s_{W})$; and (2) an indirect equilibrium effect that operates through the allocation rule $\mu_{a}(\cdot)$, which is shifted by changes in both the market-clearing parameters, $\mathbf{c}(P_{R|\theta, s_{W}})$, and the aggregate report distribution, $P_{R|\theta, s_{W}}$.

The key challenge is that Proposition \ref{prop:id_sup} only restricts the distribution of $P^{count}_{D}(\theta,s_{W})$ on the support of the distribution of the observed data, $P_{D}^{obs}$. To see why this creates a problem, note that in most economically relevant examples, a small change in $\mathbf{c}$ induces a discontinuous change in allocation for some agents. This implies that the observed data contains no direct information about the outcomes of these agents under their newly assigned allocations. The analysis in the next section introduces the key assumptions that allow us to bridge this identification gap.
\section{The Marginal Policy Effect}\label{sec:derivation}

This section derives our main identification result for the Marginal Policy Effect (MPE). The derivation must confront the central technical challenge of this environment: the inherent discontinuities in centralized allocation mechanisms. We begin by showing why these discontinuities violate the core smoothness assumptions of standard statistical methods, such as a score-based decomposition, rendering them insufficient for identifying the total MPE. To overcome this, our analysis proceeds in two steps. First, we dissect the indirect equilibrium effect into two identified economic forces: a competition effect, from shifting reports, and a market conduct effect, from the response of the clearing parameters. Second, we show how these components can be combined to construct our central result: a single, policy-invariant structural object for each agent—the equilibrium-adjusted outcome, $\Psi_i^{\text{total}}$. This variable captures an agent's full contribution to welfare, including the market externalities they generate, and ultimately reduces the MPE of any policy to a simple covariance with the policy's score. Our discussion proceeds informally to build intuition; Appendix~\ref{app:mpe_proof} collects the regularity conditions and formal proofs.

\subsection{A Score-Based Decomposition}

To analyze the MPE, we begin with a standard method for evaluating the impact of a marginal perturbation of a distribution: a score-based decomposition. If the joint density of the data, $f_D^{count}(y, a, r, w | \theta,s_{W})$, varies smoothly, the effect on welfare can be expressed as the covariance between the outcome and the model's score.\footnote{Formally, this requires the family of densities to be Differentiable in Quadratic Mean (DQM), a central concept in asymptotic statistics; see \citet{van2000asymptotic}. The interchange of differentiation and integration is permissible because the mean is a differentiable functional, subject to mild moment conditions \citep{van1991differentiable}.} The derivative of the aggregate welfare function is then given by:
\begin{align*}
    \mathcal{U}'(s_{W}) = \frac{\partial}{\partial \theta} \mathbb{E}_{(\theta,s_{W})}[Y_i] \bigg|_{\theta=0} = \mathbb{E}_{0}[Y_i s_{D}(Y_i, A_i, R_i, W_i | s_{W})],
\end{align*}
where the score, $s_{D}(Y_{i},A_i, R_i, W_i|s_{W}) = \frac{\partial}{\partial \theta} \log f_{D}^{count}(Y_{i},A_i, R_i, W_i|\theta,s_{W}) \big|_{\theta=0}$, measures the sensitivity of an observation's log-likelihood to the policy change.

This approach is powerful because it yields a highly intuitive decomposition of the total welfare effect. Based on our factorization of the data-generating process from Proposition~\ref{prop:id_sup}, the score is additive in its components, allowing us to write:
\begin{equation} \label{eq:welfare_decomp}
     \mathcal{U}'(s_{W}) = \underbrace{\mathbb{E}_{0}[Y_i s_{W}(W_i)]}_{\text{Direct Effect}} + \underbrace{\mathbb{E}_{0}[Y_i s_{A|R}(A_i|R_i, s_{W})]}_{\text{Indirect Effect}}.
\end{equation}
This decomposition, a specific application of the concepts introduced by \citet{hu2022average}, separates the MPE into two channels. The first term is a direct effect: the impact of perturbing the policy instrument $W_i$, holding the market's allocation rule fixed. The second is an indirect effect, which captures the equilibrium consequences of the policy as the allocation mechanism adjusts to the change.

The decomposition \eqref{eq:welfare_decomp}, however, rests on the critical assumption of smoothness, which, as already foreshadowed by the discussion following Proposition~\ref{prop:id_sup}, fails in the settings we study. The reason is fundamental to the nature of centralized markets: their allocation rules are often discontinuous. For instance, a school admission rule is a step function of a student's test score. A marginal policy reform that tightens admission standards by infinitesimally raising the cutoff score has no effect on most students, but it has a discrete and dramatic effect on students right at the original cutoff, who now lose their seats. This economic discontinuity breaks the mathematics behind the standard score-based machinery. Because the allocation probability, $\mu_a(r, \mathbf{c}(\theta), P_{R|\theta})$, does not change smoothly with a reform that moves a hard threshold, the allocation score, $s_{A|R}(A_i|R_i, s_{W})$, is not a well-defined random variable that we can evaluate for each agent. The standard approach, which relies on this score, cannot be directly applied.\footnote{Formally, the family is not DQM because a first-order change in the cutoff induces a first-order change in the support of the distribution of the data.} 

Despite this challenge, the decomposition in \eqref{eq:welfare_decomp} remains a valuable conceptual tool. In particular, the first term---the direct effect---is well-defined and identifiable. This is because the policymaker controls the perturbation to the policy distribution, $f(w|\theta,s_{W})$, and can ensure that the policy score $s_{W}(W_i)$ is well-behaved. For example, if $W_i \in \{0,1\}$ is a binary treatment with baseline probability $\pi_0 = \E_{0}[W_i]$, the score is proportional to $\frac{W_i}{\pi_0} - \frac{1-W_i}{1-\pi_0}$, and the direct effect becomes:
\begin{align*}
    \mathbb{E}_{0}[Y_i s_{W}(W_i)] \propto \E_{0}[Y_i|W_i = 1] - \E_{0}[Y_i|W_i = 0].
\end{align*}
Under random assignment (Assumption~\ref{ass:random_assignment}), this corresponds to the average treatment effect for the potential outcome $Y_i(w, A_i(w))$. Here $A_i(w)$ is a random variable with the distribution $\mu_a(R_i(w), \mathbf{c}_{0}, P_{R|0})$. 
This direct effect is therefore the average causal effect of the policy on the final outcome $Y_i$ in the baseline market equilibrium. 

The central technical challenge of our analysis, therefore, is to characterize and identify the indirect effect when the allocation score is ill-defined. The intuition for our approach, which we discuss in detail in the next section, is that while we cannot measure the effect of the market response on all agents, we can identify it by focusing precisely on those at the margin of their allocation---the very agents for whom an infinitesimal change in market conditions alters their allocation.

\subsection{The Indirect Effect: Competition and Market Conduct}

The indirect effect operates through the two distinct economic channels foreshadowed in the introduction. A policy reform alters the distribution of reports, changing the competitive environment. This leads to a competition effect. In response, the market's conduct adjusts the clearing parameters $\mathbf{c}$, leading to a market conduct effect. To formalize these, we define a counterfactual welfare function that depends on arbitrary clearing parameters $(\mathbf{c}, P)$:
\begin{align*}
    \mathcal{U}(\mathbf{c}, P) := \mathbb{E}_{0}\left[\sum_{a=0}^K m_a(R_i) \mu_a(R_i, \mathbf{c}, P)\right],
\end{align*}
where $m_{a}(r):= \mathbb{E}_0[Y_i(W_i,a)|R_i = r]$. Applying the chain rule to the aggregate welfare function $\mathcal{U}(\mathbf{c}_0, P_{R|0})$ decomposes the indirect effect:
\begin{align*}
 \text{Indirect Effect} = \underbrace{ \nabla_{\mathbf{c}}\mathcal{U}(\mathbf{c}_0, P_{R|0}) \cdot \mathbf{c}'[s_{R}] }_{\text{Market Conduct Effect}} + \underbrace{D_{P}\mathcal{U}(\mathbf{c}_0, P_{R|0})[ s_{R}]}_{\text{Competition Effect}}.
\end{align*}
This decomposition is driven by the report score $s_{R}(R_i) := \mathbb{E}_{0}[s_{W}(W_i)|R_i]$. The remainder of this section is dedicated to characterizing these two effects. To do so, our strategy is to first impose a structure on the allocation rule that isolates the source of the discontinuity.

\begin{assumption}[Well-Behaved Mechanism]\label{ass:mechanism_structure}
The allocation probability, $\mu_a$, can be decomposed into a smooth component, $h_a$, and a sharp eligibility boundary, $\phi_a$, with the following properties:
\begin{enumerate}
    \item[\textbf{(a)}] The allocation rule takes the form:
    \begin{equation*}
        \mu_{a}(R_i, \mathbf{c}, P) = h_a(R_{i}, \mathbf{c}, P) \ind\{ \phi_{a}(R_{i}, \mathbf{c}) \ge 0\}.
    \end{equation*}
    \item[\textbf{(b)}] The conditional allocation probability, $h_a(R_i, \mathbf{c}, P)$, is a smooth function of the market-clearing parameters, $\mathbf{c}$, and the aggregate report distribution, $P$.
    \item[\textbf{(c)}] The eligibility index, $\phi_a(R_i, \mathbf{c})$, is a smooth function of the clearing parameters, $\mathbf{c}$. Crucially, it does not depend on the aggregate report distribution $P$.
\end{enumerate}
\end{assumption}

This structure is general enough to capture a wide range of common market designs, including all examples discussed in this paper. Intuitively, it represents the allocation as a combination of a lottery and a cutoff rule. An agent must first pass a hard eligibility threshold determined by the cutoff rule ($\phi_a \ge 0$). Conditional on being eligible, they are then assigned the good with some probability determined by the lottery ($h_a > 0$). This structure ensures that while the overall competitive environment ($P$) affects allocation probabilities smoothly through the $h_a$ term, sharp discontinuities are driven solely by the interaction of the clearing parameter $\mathbf{c}$ with agent reports at the eligibility boundary.

\subsubsection{The Competition Effect}

The competition effect captures the welfare impact of the shift in the distribution of reports, holding the clearing parameters fixed. To quantify this, we consider how perturbing the density of one agent's report, $r'$, affects the allocation probability of another agent with report $r$. This peer externality is captured by a functional derivative of $h_a$ with respect to $P$, which we denote $L_{a}(r, r')$. The total competition effect is the expected impact of this change on welfare:
\begin{align*}
    D_{P}\mathcal{U}(\mathbf{c}_0, P_{R|0})[ s_{R}] = \mathbb{E}_{0}\left[\sum_{a=0}^K m_a(R_i) \mathbb{E}_{0}[L_{a}(R_i, R_j)s_{R}(R_j)|R_i]\right],
\end{align*}
where $R_j$ is an independent copy of $R_i$. This expression presents a potential identification challenge, as it depends on the unobserved conditional mean of the potential outcome, $m_a(R_i)$. However, our mechanism structure resolves this: if an agent is ineligible for good $a$, a marginal change in others' reports cannot make them eligible, meaning the spillover effect must also be zero ($L_a=0$). This feature allows us to use an inverse probability weighting approach to identify the competition effect. Since the terms in the sum are non-zero only when $\mu_a(R_i) > 0$, we can substitute the observed outcome $Y_i$ for the unobserved mean, which yields an identified expression:
\begin{align*}
    D_{P}\mathcal{U}(\mathbf{c}_0, P_{R|0})[s_{R}] = \mathbb{E}_{0}[\gamma(R_j) s_{W}(W_j)],
\end{align*}
where $\gamma(R_j) := \mathbb{E}_{0}\left[Y_i\frac{L_{A_i}(R_i, R_j)}{\mu_{A_i}(R_i, \mathbf{c}_0, P_{R|0})}|R_j\right]$. This term, $\gamma(R_j)$, is the average welfare spillover that an agent with report $R_j$ imposes on others through pure competition, holding the market's conduct fixed. Since the functional derivative $L_{A_i}$ is known from the mechanism rule (Assumption~\ref{ass:mechanism_structure}), this entire expression is identified from the data.

\subsubsection{The Market Conduct Effect}

The market conduct effect, $\nabla_{\mathbf{c}}\mathcal{U}(\mathbf{c}_0, P_{R|0}) \cdot \mathbf{c}'[s_{R}]$, captures the total welfare impact from the endogenous response of the market-clearing parameters. Characterizing it requires understanding two building blocks: first, how aggregate welfare responds to an infinitesimal change in the clearing parameters ($\nabla_{\mathbf{c}}\mathcal{U}$), and second, how the clearing parameters themselves respond to the policy reform ($\mathbf{c}'[s_{R}]$).

\paragraph{The Welfare Response to Clearing Parameters ($\nabla_{\mathbf{c}}\mathcal{U}$).} The gradient $\nabla_{\mathbf{c}}\mathcal{U}$ captures the effect of tightening or loosening the market's eligibility constraints. An infinitesimal change in $\mathbf{c}$ affects welfare through two channels: a smooth change for inframarginal agents (via $h_a$) and a discontinuous change for marginal agents at the eligibility boundary (via $\phi_a$). To formally analyze these marginal agents, we require the following regularity condition.

\begin{assumption}[Marginal Agents]\label{ass:marginal}
For each $a \in \mathcal{A}$:
\begin{enumerate}
    \item[\textbf{(a)}] The report $R_i$ consists of two components, $(R_{i,un}, R_{i,cont})$, such that conditional on $R_{i,un}$ the distribution of $R_{i,cont}$ is absolutely continuous with respect to the Lebesgue measure.
    \item[\textbf{(b)}] The functions $h_a$ and $\phi_a$ are smooth in $r_{cont}$.
    \item[\textbf{(c)}] The continuous report component has a non-degenerate effect on eligibility, 
    \begin{align*}
        \| \nabla_{r_{cont}}\phi_{a}(r_{un},r_{cont}, \mathbf{c})\|_{2} > 0.
    \end{align*}
    \item[\textbf{(d)}] The conditional mean potential outcome, $m_a(r) := \mathbb{E}_0[Y_i(W_i,a)|R_i = r]$, is continuous in $r_{cont}$.
\end{enumerate}
\end{assumption}

This assumption ensures that the concept of "marginal agents" is well-defined and provides the necessary regularity to identify the welfare impact at the boundary. The gradient $\nabla_{\mathbf{c}}\mathcal{U}$ is the sum of the effects on these two groups. Let $\Xi_a(R_i)$ be the welfare impact for agents at the margin of allocation $a$:

\begin{equation*}
 \Xi_a(R_i) := m_{a}(R_i)h_{a}(R_i, \mathbf{c}_0, P_{R|0})\nabla_{\mathbf{c}} \phi_{a}(R_i, \mathbf{c}_0).
\end{equation*}
The total gradient is then given by:
\begin{align*}
    \nabla_{\mathbf{c}}\mathcal{U}(\mathbf{c}_0, P_{R|0}) = & \underbrace{\sum_{a=0}^K \mathbb{E}_0\left[\mathbb{E}_{0}\left[\Xi_a(R_i) \bigg| \phi_{a}(R_i,\mathbf{c}_0) = 0,R_{i,un}\right]f_{\phi_{a}|R_{un}}(0|R_{i,un})\right]}_{\text{Marginal (RDD) Effect}} \\
    & + \underbrace{\mathbb{E}_{0}\left[Y_i \frac{\nabla_{\mathbf{c}} h_{A_i}(R_i, \mathbf{c}_0, P_{R|0})}{h_{A_i}(R_i, \mathbf{c}_0, P_{R|0})}\right]}_{\text{Inframarginal Effect}}
\end{align*}
Each component of this gradient is identified from the data. The inframarginal effect is a standard expectation over observed quantities, using the known functional form of $h_a$. The marginal effect is a sum of RDD effects, which are identified under Assumption~\ref{ass:marginal} from local comparisons of agents at the eligibility boundary. As we will demonstrate concretely in our examples in Section~\ref{sec:examples}, these boundary terms often correspond directly to the LATEs that are the focus of the empirical market design literature. Our framework thus clarifies that these well-studied parameters are not merely reduced-form objects but are, in fact, essential structural inputs for any equilibrium analysis. 

\paragraph{The Response of Clearing Parameters ($\mathbf{c}'[s_R]$).} The second building block, the derivative of the market conduct rule $\mathbf{c}'(\cdot)$, describes how the clearing parameters themselves respond to a policy reform. This response depends critically on the nature of the rule, and we highlight two canonical cases that correspond to distinct economic environments.

\textbf{Case 1: Competitive Equilibrium.} In many markets, the clearing parameters passively adjust to satisfy exogenous constraints, such as fixed supply. This is analogous to a competitive equilibrium where prices clear the market. As shown in our fixed-supply example from Section~\ref{sec:framework}, the conduct rule $\mathbb{E}_{P_{R}}[\mu_{a}(\mathbf{c},R_i)] = q_a$ depends on integrals over the entire report distribution. This property ensures the derivative $\mathbf{c}'[s_R]$ is a continuous linear functional in the standard space of square-integrable functions, $L_2$. This means the derivative can be represented by a familiar influence function through a standard expectation:
\begin{equation*}
    \mathbf{c}'[s_R] = \mathbb{E}_{0}[\boldsymbol{\psi}_{\mathbf{c}_0}(R_i)s_{R}(R_i)].
\end{equation*}
The influence function $\boldsymbol{\psi}_{\mathbf{c}_0}(R_i)$ is derived from the implicit function theorem and its form is determined by the local structure of the market at the baseline equilibrium:
\begin{equation*}
    \boldsymbol{\psi}_{\mathbf{c}_0}(R_i) = - \mathbf{J}_0^{-1}\begin{pmatrix}
        \mu_{1}(\mathbf{c}_0,R_i)\\ \vdots \\ \mu_{K}(\mathbf{c}_0,R_i)
    \end{pmatrix}.
\end{equation*}
Here, the matrix $\mathbf{J}_0$ is the Jacobian of the vector of aggregate market shares with respect to $\mathbf{c}$. This object, which is identified from the data given Assumption~\ref{ass:mechanism}, captures how a small change in the clearing parameters affects aggregate demand.

\textbf{Case 2: Monopoly.} In other settings, the market maker is an active agent who sets parameters to optimize an objective, such as maximizing revenue. This is analogous to a monopolist's problem. Our optimal reserve price example, where the rule $(1-F_{R}(c)) - c f_{R}(c) = 0$ is the seller's first-order condition, illustrates this case. Here, the rule's dependence on the probability density function, $f_R(c)$, at the specific point $c$ makes the operator mathematically ill-behaved in the standard $L_2$ space. Handling such cases requires restricting the analysis to smoother policy reforms, which is achieved by working in a different function space (a Sobolev space).

The contrast between these two economically distinct cases—passive market clearing versus active optimization—motivates our general approach, which we formalize next.

\begin{assumption}[Differentiability of the Market Conduct Rule]\label{ass:c_prime}
    The market conduct rule $\mathbf{c}(\cdot)$ is differentiable at $P_{R|0}$ with respect to a tangent set of scores in a Hilbert space $\mathcal{H}_{R}$. Its derivative has the representation:
    \begin{equation*}
        \mathbf{c}^{\prime}[s_{R}] = \langle \boldsymbol{\psi}_{\mathbf{c}_0}, s_{R}\rangle_{\mathcal{H}_R},
    \end{equation*}
    where $\langle \cdot , \cdot \rangle_{\mathcal{H}_R}$ is the inner product on $\mathcal{H}_{R}$ and $\boldsymbol{\psi}_{\mathbf{c}_0}$ is the representer of the derivative.
\end{assumption}
The function $\boldsymbol{\psi}_{\mathbf{c}_0}(R_i)$ is the influence function of the market conduct rule, generalized to the appropriate space $\mathcal{H}_R$.

\subsection{The Equilibrium-Adjusted Outcome}
Having established that each building block of the indirect effect---the competition externality $\gamma(R_i)$, the welfare gradient $\nabla_{\mathbf{c}}\mathcal{U}$, and the influence function of the market conduct rule $\boldsymbol{\psi}_{\mathbf{c}_0}$---is identified from the data under our assumptions, we can now combine them to state our main result. The total MPE is the sum of the direct effect and the two components of the indirect effect. The technical subtleties of the market conduct rule prevent the entire MPE from always being expressed as a single covariance. Our main theorem, therefore, presents the MPE in a more general form that respects this distinction.

\begin{theorem}[The Marginal Policy Effect]\label{th:mpe_general}
Suppose Assumptions \ref{ass:mechanism}-\ref{ass:structural_invariance} and \ref{ass:mechanism_structure}-\ref{ass:c_prime} hold and for each $a \in \mathcal{A}$ the conditional expectation function, $m_{a}(r_{un},r_{cont})$, is bounded and continuous in $r_{cont}$. Then, the MPE is identified and can be expressed as:
\begin{align*}
\mathcal{U}'(s_W) = \mathbb{E}_{0}[\Psi_i^{\text{fixed}} s_W(W_i)] + \langle \Psi^{\text{conduct}}, s_R \rangle_{\mathcal{H}_R}
\end{align*}
where $\Psi_i^{\text{fixed}}$ represents the portion of an agent's welfare contribution independent of the market's conduct:
\begin{equation*}
    \Psi_i^{\text{fixed}} = \underbrace{Y_i}_{\text{Private Outcome}} + \underbrace{\gamma(R_i)}_{\text{Competition Externality}},
\end{equation*}
and $\Psi^{\text{conduct}}$ captures the market conduct externality:
\begin{equation*}
 \Psi^{\text{conduct}}(R_i) := \nabla_{\mathbf{c}}\mathcal{U}(\mathbf{c}_0, P_{R|0}) \cdot \boldsymbol{\psi}_{\mathbf{c}_0}(R_i).
\end{equation*}
\end{theorem}

Theorem~\ref{th:mpe_general} provides a universally applicable formula for the MPE. Its two-part structure cleanly separates the welfare change into components that can be analyzed using standard covariance-based methods and a component that depends on the specific geometry of the policy space, $\mathcal{H}_R$. This general form simplifies if the market conduct rule is differentiable in the standard $L_2$ space.

\begin{corollary}[The Equilibrium-Adjusted Outcome]\label{cor:psi}
If the market conduct rule $\mathbf{c}(P)$ is differentiable in $\mathcal{H}_R = L_2(R_i)$, then the MPE from Theorem~\ref{th:mpe_general} can be written as:
\begin{align*}
\mathcal{U}'(s_{W}) = \mathbb{E}_{0}[\Psi_i^{\text{total}} s_W(W_i)],
\end{align*}
where $\Psi_i^{\text{total}}$ is the score-independent equilibrium-adjusted outcome. It represents an agent's total contribution to welfare:
\begin{equation*}
    \Psi_i^{\text{total}} = \Psi_i^{\text{fixed}} + \Psi^{\text{conduct}}(R_i).
\end{equation*}
\end{corollary}
This corollary recovers the powerful intuition from our motivating discussion. In many common environments, all complex market interactions can be summarized by a single, policy-invariant structural object, $\Psi_i^{\text{total}}$. This object represents the correct welfare-relevant outcome for a policymaker. To find the welfare-maximizing local reform, one must simply find the policy score that has the highest covariance with this fixed, structural outcome.

The power of this "separation principle" is that it provides a unified foundation for addressing a range of practical policy questions. By isolating the full market structure in the single object $\Psi_i^{\text{total}}$, our framework provides a flexible tool for applied work. As we demonstrate in Section~\ref{sec:extensions}, this allows our framework to address complex empirical challenges, including optimal policy targeting and endogenous selection.
\section{Examples}\label{sec:examples}

In this section, we illustrate how our general framework applies to several canonical economic environments. These examples serve to build intuition by showing how the abstract components of the marginal policy effect map onto concrete, estimable quantities in specific models. The examples also clarify the conditions under which a researcher can identify the full welfare function versus only its local gradient. We begin with simple single-product markets and proceed to more complex, multi-product settings.

\subsection{Single Product with Random Rationing}

Consider a market for a single product ($\mathcal{A}=\{0,1\}$) with fixed supply $q$. The policy is a binary treatment, $W_i \in \{0,1\}$, and agents submit a preference report $R_i \in \{0,1\}$, where $R_i=1$ indicates a desire for the product. The product is allocated via random rationing, so the allocation probability is 
\begin{align*}
    \mu_1(R_i, c) = R_i \cdot c,
\end{align*}
where the rationing probability $c$ is the market-clearing parameter that adjusts to satisfy the fixed supply constraint
\begin{align*}
    \mathbb{E}_{(\theta,s_{W})}[\mu_1(R_i,c)] = q_1.
\end{align*}
The mechanism is strategy-proof, so this report reflects agents' true preferences. We assume excess demand at baseline, $\E_{0}[R_i] > q_1$.  The equilibrium constraint $\E_{(\theta,s_W})[R_i \cdot c] = q_1$ implies that $c(P_{R|\theta,s_W}) = q_1 / \E_{(\theta,s_W)}[R_i]$. A policy reform alters the share of agents demanding the product, which creates an equilibrium effect through the adjustment of $c(\theta)$.

We derive the equilibrium-adjusted outcome $\Psi_i^{\text{total}}$ by applying Corollary~\ref{cor:psi}. Since the allocation rule $\mu_1(R_i, c)$ depends only on an agent's own report and the clearing parameter $c$, and not on the aggregate report distribution $P_R$, the competition term, $\gamma(R_i)$, is zero. The indirect effect therefore operates entirely through the market conduct externality. Let $\tau(r) := \mathbb{E}_{\theta_0}[Y_i(W_i, 1) - Y_i(W_i, 0)|R_i = r]$ denote the average treatment effect of the allocation for agents with report $r$. The market conduct externality an agent imposes by demanding the good ($R_i=1$) simplifies to $\tau(1)$.

The resulting equilibrium-adjusted outcome is:
\begin{equation*}
    \Psi_i^{\text{total}} = Y_i - \tau(1) \cdot A_i.
\end{equation*}
The interpretation is direct: an agent's total contribution to welfare is their private outcome, $Y_i$, net of the externality they impose by receiving a unit of the scarce good ($A_i =1$). This externality is valued at $\tau(1)$, which represents the average welfare gain the good provides to the other potential recipients displaced at the margin.

This stylized example connects directly to empirical work. The term $\tau(1)$ is a policy-specific local average treatment effect. If the policy $W_i$ has no direct effect on outcomes ($Y_i(w,a) = Y_i(a)$), this term simplifies to the standard LATE identified in school choice lotteries (e.g., \citealp{abdulkadirouglu2017research, walters2018demand}). Our framework shows that this familiar estimand is not just a reduced form quantity; it is a structural object needed to conduct counterfactual policy analysis, echoing the approach in \citet{kline2016evaluating}.

\subsection{Price-Based Allocation}

We now consider a market where the report $R_i \in \mathbb{R}_{+}$ is a continuous valuation for a single product, and allocation is determined by a market-clearing price or cutoff, $c$. An agent receives the product if their valuation exceeds the price, so the allocation rule is the discontinuous function $\mu_{1}(R_i,c)= \ind\{R_i > c\}$. We assume the distribution of reports $R_i$ admits a continuous, positive density, $f(r)$. We also assume that the conditional mean functions $\{m_0(r), m_1(r)\}$ are continuous functions of reports. The market-clearing price $c(P_{R|\theta, s_W})$ is set to satisfy the supply constraint,
\begin{align*}
    \mathbb{E}_{(\theta,s_{W})}[\ind\{R_i > c\}] = q,
\end{align*}
which implies that $c(P_{R|\theta, s_W})$ is the $(1-q)$-quantile of the report distribution under policy regime $\theta$.

As in the random rationing case, the allocation rule does not depend on the aggregate report distribution $P_R$, so the competition externality term, $\gamma(R_i)$, is zero. The indirect effect operates entirely through the market conduct externality. The welfare gradient with respect to the cutoff is the effect on aggregate welfare of marginally raising the price, which under stated conditions is
\begin{align*}
    \nabla_{c}\mathcal{U}(\theta_0) = -[m_1(c_0) - m_0(c_0)]f(c_0).
\end{align*}
The market's response, $c'(\cdot)$ is given by the following influence function
\begin{equation*}
    \psi_{c_0}(R_i) = -\frac{\ind\{R_i > c\}}{f(c_0)} =  -\frac{A_i}{f(c_0)}
\end{equation*}
The density terms cancel, leaving a simple expression for the market conduct term:
\begin{align*}
    \Psi_i^{\text{total}} = Y_i - \tau(c_0) \cdot A_i.
\end{align*}
The simplicity of this final expression, where the density terms cancel, reveals a powerful economic intuition. The market conduct externality is the product of two opposing forces. The first is the aggregate welfare impact of a marginal increase in the cutoff, which is large when the density at the cutoff, $f(c_0)$, is high, as many agents are affected. The second is the influence of a single inframarginal agent on the equilibrium cutoff, which is small when the density $f(c_0)$ is high, as only a small price change is needed to displace one marginal agent to make room for them. These two effects, one proportional to the density and the other inversely proportional to it, exactly offset each other. The result is that the externality any inframarginal agent imposes is simply the welfare loss of the single agent at the margin that they displace, $-\tau(c_0)$. This last term, $\tau(c_0)$, is precisely the RDD estimand used to quantify the effects of charter schools (e.g., \citealp{abdulkadi̇rouglu2022breaking}). Our framework demonstrates that this RDD parameter can be directly used to compute the welfare consequences of any local policy.

\subsection{Second-Price Auction}

We now illustrate the full decomposition of the MPE in a second-price auction for a single good with $n$ i.i.d. participants. The second-price auction is strategy proof, which implies that bidding one's private valuation, $R_i \in \mathbb{R}_{+}$, is a dominant strategy. The platform sets a reserve price, $c$, to ensure the ex-ante probability of winning is a fixed quantity, $q$.\footnote{This is relevant in applications like sponsored search, where a platform may wish to display advertisements with a certain frequency.} An agent with valuation $R_i$ wins if they bid above the reserve price and have the highest bid among all participants, so their win probability is 
\begin{align*}
    \mu_1(R_i, c, P_R) = \ind\{R_i > c\} \cdot [F_{R}(R_i)]^{n-1}.
\end{align*} 

A policy that perturbs the distribution of valuations creates spillovers through two distinct channels. First, it affects the reserve price $c(P_{R|\theta, s_W})$ needed to meet the win-rate target---a market conduct effect. Second, it changes the distribution of competing bids, $P_{R|\theta,s_W}$, altering the win probability for all bidders above the reserve price. As a result, this example features a non-zero competition effect, $\gamma(R_i) \neq 0$.

To see how this competition effect is constructed, we first compute the functional derivative of the win probability, $L_1(r, r')$. A change in the density of bidders at value $r'$ only affects bidders with valuations $r > r'$, as it changes the value of the CDF $F_R(r)$. Applying the definition of the functional derivative yields:
\begin{equation*}
    L_1(r, r') = \frac{\partial \mu_1(r, c, P_R)}{\partial P_R(r')} = \ind\{r > c\} \cdot (n-1)[F_R(r)]^{n-2} \cdot \ind\{r' \le r\}.
\end{equation*}
Substituting this into the general formula for the competition externality from Section~\ref{sec:derivation}, $\mathbb{E}_{0}\left[Y_i\frac{L_{A_i}(R_i, R_j)}{\mu_{A_i}(R_i, \mathbf{c}_0, P_{R|0})}|R_j\right]$ gives the expression for $\gamma(R_i)$.

Combining these components yields the equilibrium-adjusted outcome, which is the sum of the private outcome and two distinct externality terms:
\begin{align*}
    \Psi_i^{\text{total}} = \underbrace{Y_i}_{\text{Private Outcome}} + \underbrace{\gamma(R_i)}_{\text{Competition Externality}} + \underbrace{\Psi^{\text{conduct}}(R_i)}_{\text{Market Conduct Externality}}.
\end{align*}
Here, the competition externality, $\gamma(R_i)$, is the welfare impact an agent's bid imposes on inframarginal competitors. It can be expressed intuitively using the maximum order statistic of the competing bids, $R_{(n-1)}$:
\begin{align*}
    \gamma(R_i) = \mathbb{E}_0\left[\tau(R_{(n-1)}) | R_{(n-1)} \ge \tilde{r}\right] \times \left(1 - F_{R|0}(\tilde{R}_i)^{n-1}\right), \quad \text{where } \tilde{R}_i = \max(c_0, R_i).
\end{align*}
This is the expected treatment effect for the winning competitor, conditional on them being a relevant threat (bidding above $\tilde{r}$), multiplied by the probability that such a threat exists. The market conduct externality, $\Psi^{\text{conduct}}(R_i)$, is the welfare impact from agent $i$'s influence on the equilibrium reserve price:
\begin{align*}
\Psi^{\text{conduct}}(R_i) = -\tau(c_0) F_{R|0}(c_0)^{n-1} \ind\{R_i > c_0\}.
\end{align*}

\subsubsection{Optimal reserve price}
Our analysis above focused on a $c(\cdot)$ that forces a fixed allocation probability. We now discuss a  more complex objective for the platform: setting the reserve price $c$ to maximize the expected revenue, following the principles of optimal auction design \citep{myerson1981optimal}. The optimal reserve price $c(P_R)$ is the solution to the first-order condition:
\begin{equation*}
    1 - F_{R}(c(P_R)) - c(P_R) f_{R}(c(P_R)) = 0.
\end{equation*}
As discussed in our theoretical section, the dependency on the density $f_R$ means the derivative $c'(P_R)$ is not a continuous operator in $L_2$. Correctly characterizing the market conduct effect requires our general framework based on a Sobolev space, $\mathcal{H}_R$, with the inner product:
\begin{equation*}
    \langle \psi_{c_0}, s_R \rangle_{\mathcal{H}_R} := \mathbb{E}_{0}[\psi_{c_0}(R_i)s_R(R_i)] + \mathbb{E}_{0}[\psi_{c_0}'(R_i)s_R'(R_i)].
\end{equation*}
As shown in Appendix~\ref{app:examples}, the representer for the derivative of the optimal reserve price, $\psi_{c_0}$, is the solution to the following Sturm-Liouville differential equation:
\begin{equation*}
     \psi_{c_0}(r)f_{R|0}(r) - (\psi_{c_0}'(r)f_{R|0}(r))' = K \cdot \left( \ind\{r \le c_0\}f_{R|0}(r) + \alpha \delta(r-c_0) \right).
\end{equation*}
While complex, this equation can be solved for a given baseline density $f_{R|0}$, yielding the influence function $\psi_{c_0}(\cdot)$. The MPE is then fully identified by applying our general formula from Theorem~\ref{th:mpe_general}, demonstrating the framework's capacity to handle a wide class of economically relevant mechanisms.

\subsection{School Choice with Multiple Schools}

Our final example is a multi-product extension of the price-based model: a centralized school choice mechanism. As shown by \citet{azevedo2016supply}, allocations in large matching markets can often be characterized by a vector of market-clearing score cutoffs, which makes this an empirically relevant setting.

Consider a market with two schools ($k=1,2$) and an outside option ($k=0$), each with capacity $q_k$. A student's type $R_i$ consists of a preference ranking, $\succ_i$, and a vector of school-specific scores, $(V_{i,1}, V_{i,2})$. A student is assigned to their most-preferred school $k$ for which they are eligible, which requires their score to exceed the school's cutoff, $V_{i,k} > c_k$. The vector of cutoffs $\mathbf{c}=(c_1, c_2)$ is set endogenously to ensure that the number of assigned students exactly meets the capacity constraints for each school.

A policy reform perturbs the joint distribution of preferences and scores. To maintain equilibrium, the market responds by adjusting the cutoff vector by a marginal amount, $\mathbf{c}'=(c_1', c_2')$. This change in cutoffs reallocates students who are precisely on the margin of admission. The score space, shown in Figure~\ref{fig:scorespace}, helps build intuition. The initial cutoffs $(c_1, c_2)$ define eligibility regions. When the policy changes, the cutoffs shift, creating thin ``bands'' of students whose eligibility status changes. The welfare impact of the reform depends entirely on who these marginal students are and how they are reallocated based on their preferences.

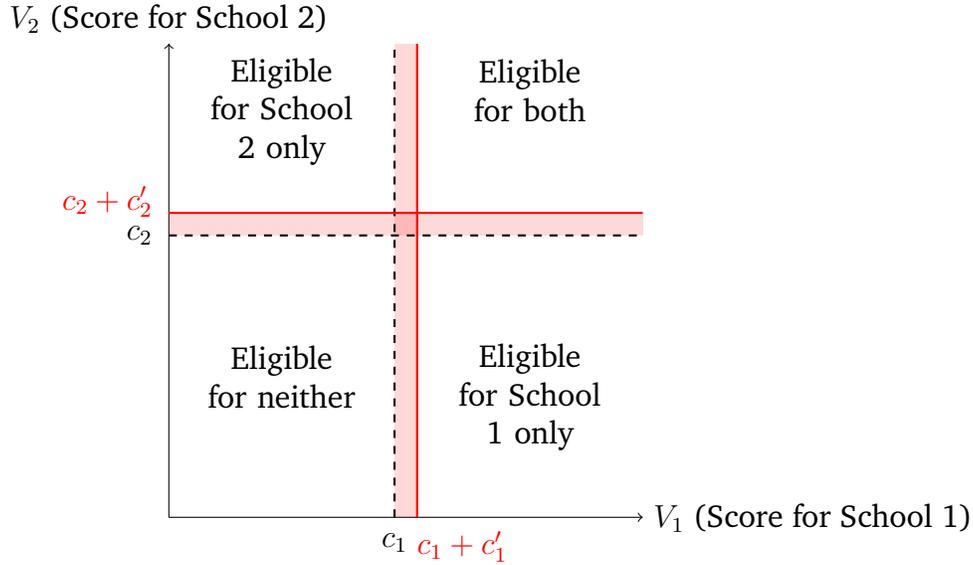
\begin{figure}[t!]
\centering
\begin{tikzpicture}[scale=1.5]
    \draw[->] (0,0) -- (4.2,0) node[right] {$V_1$ (Score for School 1)};
    \draw[->] (0,0) -- (0,4.2) node[above] {$V_2$ (Score for School 2)};
    \def\cOne{2}
    \def\cTwo{2.5}
    \def\cOnePrime{2.2}
    \def\cTwoPrime{2.7}
    \fill[red, opacity=0.15] (\cOne,0) rectangle (\cOnePrime,4.2);
    \fill[red, opacity=0.15] (0,\cTwo) rectangle (4.2,\cTwoPrime);
    \draw[dashed, thick] (\cOne,0) -- (\cOne,4.2);
    \node[below=2pt] at (\cOne, 0) {$c_1$};
    \draw[dashed, thick] (0,\cTwo) -- (4.2,\cTwo);
    \node[anchor=east, left=2pt] at (0, \cTwo) {$c_2$};
    \draw[red, thick] (\cOnePrime,0) -- (\cOnePrime,4.2);
    \node[below=1pt, red] at (2.6, 0) {{$c_1+c'_1$}};
    \draw[red, thick] (0,\cTwoPrime) -- (4.2,\cTwoPrime);
    \node[anchor=east, left=2pt, red] at (0, 2.8) {{$c_2+c'_2$}};
    \node[text width=2cm, align=center] at (3.2, 3.79) {Eligible for both};
    \node[text width=2cm, align=center] at (1, 3.6) {Eligible for School 2 only};
    \node[text width=2cm, align=center] at (3.2, 1.065) {Eligible for School 1 only};
    \node[text width=2cm, align=center] at (1, 1.25) {Eligible for neither};
\end{tikzpicture}
\caption{A policy reform shifts cutoffs from the dashed to the red lines. Students in the shaded red bands are ``marginal''---their eligibility changes.}
\label{fig:scorespace}
\end{figure}

We can define the key building blocks for this effect: let $\rho_{j \to k}$ be the density of students at the cutoff for school $j$ (i.e., with score $V_{i,j}=c_j$) who, upon losing eligibility for $j$, are reallocated to school $k$. Let $\tau_{j \to k}(c_j)$ be the average causal effect of this switch for this specific group:
\begin{align*}
    \tau_{j \to k}(c_j) = \E[Y_i(W_i,k) - Y_i(W_i, j) | V_{i,j}=c_j, \text{reallocated from } j \text{ to } k].
\end{align*}
These are precisely the types of parameters estimated in RDD-based studies of school choice. Table~\ref{table:school_details} illustrates the primary reallocations (assuming cutoffs rise).

These densities of marginal students are the building blocks of the Jacobian matrix, $\mathbf{J}$, which describes the derivative of the market conduct rule, $\mathbf{c}^{\prime}(P_{R|0})$. An element $J_{kj}$ of this matrix represents the change in enrollment at school $k$ from a marginal increase in the cutoff for school $j$. The diagonal elements are negative (raising a school's cutoff lowers its own enrollment), while the off-diagonal elements are positive (raising one school's cutoff pushes some students to the other school). The Jacobian for this market is:
\begin{equation*}
\mathbf{J} = 
\begin{pmatrix}
-(\rho_{1 \to 0} + \rho_{1 \to 2}) & \rho_{2 \to 1} \\
\rho_{1 \to 2} & -(\rho_{2 \to 0} + \rho_{2 \to 1})
\end{pmatrix}
\end{equation*}
The non-zero off-diagonal terms, $\rho_{1 \to 2}$ and $\rho_{2 \to 1}$, explicitly measure the cross-school substitution effects. 

The market conduct externality for a seat at each school, which we denote by the vector $\mathbf{v}=(v_1, v_2)$, is a combination of the marginal effects at both cutoffs, adjusted for the full matrix of equilibrium interactions. Define $G_1 := - [\nabla_{\mathbf{c}}\mathcal{U}(c_0)]_1 = -(\rho_{1\to 0}\tau_{1\to 0} + \rho_{1\to 2}\tau_{1\to 2})$ and $G_2 := - [\nabla_{\mathbf{c}}\mathcal{U}(c_0)]_2 = (\rho_{2\to 1}\tau_{2\to 1} + \rho_{2\to 0}\tau_{2\to 0})$ be the total welfare effect at each margin. Solving the system of equilibrium interactions yields the following expressions for the social externality values:
\begin{align*}
    v_1 &= \frac{1}{\det(\mathbf{J})} \left[ (\rho_{2 \to 0} + \rho_{2 \to 1}) G_1 + \rho_{1 \to 2} G_2 \right] \\
    v_2 &= \frac{1}{\det(\mathbf{J})} \left[ (\rho_{1 \to 0} + \rho_{1 \to 2}) G_2 + \rho_{2 \to 1} G_1 \right]
\end{align*}
The crucial feature of these expressions is that the social value of a seat at School 1 ($v_1$) explicitly depends on the treatment effects at the margin for School 2 (embedded in $G_2$), weighted by the substitution patterns. 

The final equilibrium-adjusted outcome for a student is:
\begin{equation*}
    \Psi_i^{\text{total}} = Y_i - v_1 \cdot \mathbf{1}\{A_i=1\} - v_2 \cdot \mathbf{1}\{A_i=2\}
\end{equation*}
This example shows precisely how the framework synthesizes readily interpretable RDD treatment effects ($\tau_{j \to k}$) with the market's underlying substitution patterns ($\mathbf{J}$) to construct the policy-invariant parameters ($v_k$) required for any counterfactual policy evaluation.

\begin{table}[t!]
\centering
\caption{Classification of Marginal Reallocations}
\begin{tabular}{@{}lll@{}}
\toprule
\textbf{Marginal Group} & \textbf{Reallocation Path} & \textbf{Welfare Effect Component} \\ \midrule
$V_{i,1} \approx c_1$, Pref: $1 \succ 0$ & School 1 $\to$ Outside Option & $\rho_{1 \to 0} \cdot \tau_{1 \to 0}(c_1)$ \\
$V_{i,2} \approx c_2$, Pref: $2 \succ 0$ & School 2 $\to$ Outside Option & $\rho_{2 \to 0} \cdot \tau_{2 \to 0}(c_2)$ \\
$V_{i,1} \approx c_1$, Pref: $1 \succ 2 \succ 0$ & School 1 $\to$ School 2 & $\rho_{1 \to 2} \cdot \tau_{1 \to 2}(c_1)$ \\
$V_{i,2} \approx c_2$, Pref: $2 \succ 1 \succ 0$ & School 2 $\to$ School 1 & $\rho_{2 \to 1} \cdot \tau_{2 \to 1}(c_2)$ \\
\bottomrule
\end{tabular}\label{table:school_details}
\end{table}

\subsection{Discussion}
Our main result identifies the Marginal Policy Effect---the local gradient of the welfare function at the observed equilibrium. A natural question is under what conditions a researcher can go beyond this local result to evaluate large-scale, or "global," policy changes. Proposition~\ref{prop:id_sup} shows that global identification hinges on a stringent support condition, which requires that a policy reform does not assign agents to allocations they could never have received in the baseline equilibrium.

As our examples illustrate, however, this condition is the exception rather than the rule. It holds in markets with pervasive randomness, like the random rationing mechanism, where the allocation process itself acts as an experiment that reveals the distribution of potential outcomes \citep{narita2021algorithm}. In contrast, markets with deterministic cutoffs---such as price-based allocation, auctions, and school choice systems---violate this condition. A marginal change in a cutoff pushes agents across a sharp boundary, meaning the causal effect of the allocation is only ever revealed for agents at that specific, observed margin.

The local nature of our identification result has an immediate implication for any analysis aiming to evaluate global reforms: such an analysis must rely on extrapolation. Our framework contributes to this goal by providing a sharp delineation between what is identified from the data and what must be assumed. By first constructing the equilibrium-adjusted outcome, $\Psi_i^{\text{total}}$, applied researchers can isolate the identified foundation upon which transparent extrapolation assumptions---about functional forms or the outcomes of inframarginal agents---can be built.
\section{Applications and Extensions}\label{sec:extensions}

Section~\ref{sec:derivation} developed our main theoretical result, the "separation principle," which hinges on the construction of the equilibrium-adjusted outcome, $\Psi_i^{\text{total}}$. We now demonstrate the framework's flexibility and breadth by extending it in four directions. First, we generalize the welfare criterion beyond simple averages to a broad class of distributional objectives, such as quantiles and inequality measures. Second, we incorporate observable covariates to handle selection on observables and lay the groundwork for optimal policy targeting. Third, we address endogenous selection by connecting our framework to the MTE literature. Finally, we discuss the significant identification challenges that arise in non-strategy-proof mechanisms where agent reports are themselves endogenous.

\subsection{Beyond Average Outcomes: General Welfare Functionals}
\label{sec:general_welfare}

Our analysis has thus far defined aggregate welfare as the average outcome, $\mathcal{U}_0 = \mathbb{E}_0[Y_i]$. However, as foreshadowed by Remark~\ref{rem:in_local}, because our approach is local---focused on identifying the marginal effect of a reform---it can be extended to any welfare criterion that is sufficiently smooth with respect to the distribution of outcomes. This allows policymakers to evaluate reforms based not only on their average effects but also on their impact on other distributional objectives.

Let the welfare criterion be a functional $\mathcal{U}(F_Y)$ that maps the CDF of the outcome, $F_Y$, to a real number. The key condition for our analysis to apply is that this functional must be Hadamard differentiable at the baseline outcome distribution, $F_{Y|0}$. This is a standard smoothness condition in statistics that guarantees the existence of a well-behaved and identifiable influence function, $\text{IF}(y; F_{Y|0})$, which characterizes the marginal contribution of an observation $y$ to the overall functional.  The entire analysis from our main theorems holds, with one simple substitution: the individual outcome $Y_i$ is replaced by its marginal contribution to the welfare functional, $\text{IF}(Y_i; F_{Y|0})$. 

The new equilibrium-adjusted outcome, which we denote $\Psi_i^\mathcal{U}$, is therefore:
\begin{equation*}
    \Psi_i^\mathcal{U} = \text{IF}(Y_i; F_{Y|0}) + \gamma^\mathcal{U}(R_i) + \Psi^{\text{conduct},\mathcal{U}}(R_i),
\end{equation*}
where the competition externality ($\gamma^\mathcal{U}$) and the market conduct externality ($\Psi^{\text{conduct},\mathcal{U}}$) are constructed exactly as before, but using the conditional expectation of the influence function, $\mathbb{E}[\text{IF}(Y_i; F_{Y|0}) | A_i=a, R_i=r]$, in place of the conditional mean of the outcome.

This generalization covers a wide range of common welfare criteria.
\begin{itemize}
    \item \textbf{Quantiles.} If the policymaker is interested in the effect on the $\tau$-th quantile of the outcome distribution, $q_\tau$, the relevant influence function is $\text{IF}(y; F_{Y|0}) = \frac{\tau - \mathbf{1}\{y \le q_\tau\}}{f_{Y|0}(q_\tau)}$, where $f_{Y|0}$ is the baseline density of the outcome. Our framework can thus be used to find the MPE of any local policy on, for example, the median outcome.
    \item \textbf{Inequality Measures.} As another distributional measure, one could use the Gini coefficient. This functional is also Hadamard differentiable, and its well-known influence function can be substituted into our formulas to find the MPE of a policy on inequality.
\end{itemize}
This extension demonstrates that our framework provides a general toolkit for evaluating the local effects of policies on any social objective that can be expressed as a smooth functional of the outcome distribution.

\subsection{Covariates, Identification, and Optimal Targeting}
\label{sec:covariates}

In many empirical settings, the assumption of unconditional random assignment is unrealistic. It is often more plausible to assume selection on observables, where a policy is randomly assigned only after conditioning on a rich set of pre-determined covariates. This section extends our baseline analysis to incorporate such covariates, serving three critical purposes. First, doing so strengthens the credibility of the underlying identification assumptions. Second, it provides the necessary foundation for designing and evaluating targeted policies. Third, this extension lays the groundwork for our analysis of selection on unobservables in the subsequent section.

We generalize our baseline framework by replacing Assumption~\ref{ass:random_assignment} with the following standard condition:
\begin{assumption}[Unconfoundedness]
\label{ass:cond_random}
The policy instrument $W_i$ is randomly assigned conditional on a vector of observed, pre-determined covariates $X_i$.
\end{assumption}
Under Assumption~\ref{ass:cond_random}, a policy reform is a change to the conditional distribution of the policy instrument, characterized by a conditional score $s_{W|X}(w|x)$. This in turn induces a marginal score on the distribution of reports, $s_{R}(R_i) := \mathbb{E}_{0}[s_{W|X}(W_i|X_i)|R_i]$. The key insight of this section is that the fundamental structure of our main result remains intact. The MPE is still given by the expression from Theorem~\ref{th:mpe_general}:
\begin{equation*}
    \mathcal{U}^{\prime}(s_{W|X}) =  \mathbb{E}_{0}[\Psi_i^{\text{fixed}} s_{W|X}(W_i|X_i)] + \langle \Psi^{\text{conduct}}, s_R \rangle_{\mathcal{H}_R},
\end{equation*}
where the structural components, $\Psi_i^{\text{fixed}}$ and $\Psi^{\text{conduct}}$, are the same as those defined previously.

The invariance of these structural components might seem surprising, but it is a direct consequence of the mechanism's design. Because the allocation rule responds only to an agent's report $R_i$, and not directly to the covariates $X_i$, the equilibrium adjustment functions are anonymous with respect to this observed heterogeneity.\footnote{This does not preclude covariates from being part of the report, i.e., $X_i \subset R_i$, provided they are components that are unaffected by the policy instrument $W_i$.} This anonymity has powerful simplifying implications for empirical analysis. For instance, the local RDD estimands discussed in Section~\ref{sec:examples} 
only need to be identified unconditionally, rather than conditional on the full vector $X_i$, thereby avoiding a curse of dimensionality.

Conceptually, this implies a departure from the standard ``condition-then-aggregate'' approach often used in settings with selection on observables. Our framework instead justifies a direct aggregate analysis. This insight will be crucial in the next section, where we extend this logic to handle selection on unobservables.

This structure allows us to turn to the problem of optimal policy design. Focusing on the important class of environments from Corollary~\ref{cor:psi}, where the MPE simplifies to a single covariance, we have:
\begin{equation*}
    \mathcal{U}^{\prime}(s_{W|X}) = \mathbb{E}_0[\Psi_i^{\text{total}} s_{W|X}(W_i|X_i)].
\end{equation*}
This representation of the MPE as a linear functional connects our framework directly to the literature on optimal policy targeting and empirical welfare maximization (EWM) (e.g., \citet{manski2004statistical,kitagawa2018should,athey2021policy}).
In those frameworks, the objective is to choose a policy that maximizes the expectation of a welfare-relevant outcome. Our central result shows that in an equilibrium environment, the correct welfare-relevant object is not the observed outcome $Y_i$, but the equilibrium-adjusted outcome, $\Psi_i^{\text{total}}$. The policymaker's problem is thus to choose a targeting rule---represented by the conditional score $s_{W|X}(W_i|X_i)$---that maximizes the covariance with this fixed, structural outcome.

Our focus on the MPE as the key object for policy improvement also connects our work to a design-based literature on EWM. For instance, \citet{viviano2024policy} propose an experimental design that uses ``local perturbations'' to treatment probabilities across large, independent clusters to directly estimate the MPE in settings with unknown, decentralized spillovers. Whereas their approach provides an experimental method for estimating the total welfare gradient, our complementary framework provides its structural decomposition in centralized markets.

To make this connection explicit, we analyze the canonical case of a binary policy, $W_i \in \{0,1\}$. A targeted local reform is a marginal perturbation to the baseline propensity score, $p(X_i) = \Pr(W_i=1|X_i)$, in a direction defined by a square-integrable function $h(X_i)$. The score for such a reform is given by:\footnote{This score arises from a perturbation of the log-odds ratio, a standard way to ensure perturbed probabilities remain in $(0,1)$. Specifically, if the new log-odds is $\log\frac{p(X_i)}{1-p(X_i)} + \theta h(X_i)$, the derivative of the log-likelihood with respect to $\theta$ at $\theta=0$ yields this score.}
\begin{equation*}
    s_{W|X}(W_i|X_i) = \left(\frac{W_i}{p(X_i)} - \frac{1-W_i}{1-p(X_i)}\right) h(X_i).
\end{equation*}

Substituting this score into the MPE formula and applying the law of iterated expectations yields:
\begin{align*}
    \mathcal{U}'(s_{W|X}) = \mathbb{E}_{0}\left[ \left( \mathbb{E}_{0}[\Psi^{\text{total}}_i | W_i=1, X_i] - \mathbb{E}_{0}[\Psi^{\text{total}}_i | W_i=0, X_i] \right) \cdot h(X_i) \right]. 
\end{align*}
This result provides a clear recipe for policy design. The welfare gain from a local reform is the inner product of the Conditional Average Treatment Effect on the equilibrium-adjusted outcome (CATE-$\Psi$), defined as the term in parentheses, and the function $h(X_i)$ that defines the reform's direction. Unlike in the global EWM literature, where a policy rule maps covariates to probabilities, the function $h(X_i)$ is unconstrained in sign. It represents the gradient of the reform; a negative value for a subpopulation simply implies that the welfare-improving direction is to locally reduce their probability of treatment.

The optimal local reform is the one that maximizes this welfare gain for a given budget or ``size.'' A natural choice is to constrain the variance of the perturbation, $\mathbb{E}[h^2(X_i)] \le C$. The problem of maximizing the MPE subject to this constraint is a standard Hilbert space projection problem, whose solution, by the Cauchy-Schwarz inequality, is to set $h(X_i)$ proportional to the CATE-$\Psi$. This yields the optimal score:
\begin{equation*}
    s^{\star}_{W|X}(W_i|X_i) \propto \left(\frac{W_i}{p(X_i)} - \frac{1-W_i}{1-p(X_i)}\right)\left( \mathbb{E}_{0}[\Psi^{\text{total}}_i | W_i=1, X_i] - \mathbb{E}_{0}[\Psi^{\text{total}}_i | W_i=0, X_i] \right).
\end{equation*}
This policy for local improvement differs fundamentally from the globally optimal rule derived in the EWM literature, which typically takes the form \citep{manski2004statistical}: 
\begin{equation*}
    s^{\text{EWM}}_{W|X}(W_i|X_i) \propto \left(\frac{W_i}{p(X_i)} - \frac{1-W_i}{1-p(X_i)}\right)\mathbf{1}\left\{ \mathbb{E}_{0}[\Psi^{\text{total}}_i | W_i=1, X_i] - \mathbb{E}_{0}[\Psi^{\text{total}}_i | W_i=0, X_i] \ge 0 \right\}.
\end{equation*}
Our approach identifies the most welfare-improving direction for a marginal reform from the current baseline, which leverages the magnitude of the CATE-$\Psi$. The EWM approach, in contrast, identifies the optimal policy level within a particular class of rules, which depends only on the sign of the CATE-$\Psi$.

\begin{remark}[The Role of Covariates in Defining Reforms]
Our analysis of targeting has focused on using covariates $X_i$ to generate a rich space of policy reforms. This is particularly crucial for binary policies. When the policy instrument is binary ($W_i \in \{0,1\}$), the space of valid scores is one-dimensional; any score must be proportional to $\frac{W_i - \mathbb{E}_0[W_i]}{\mathbb{V}[W_i]}$. Without covariates, there is therefore only a single direction for local policy improvement. In contrast, a non-binary policy instrument, such as a continuous subsidy, naturally admits a high-dimensional space of reforms even without covariates. In that case, the score can be any function of the instrument, $s(W_i)$, that is orthogonal to a constant (i.e., satisfies $\mathbb{E}[s(W_i)]=0$). This allows for a wide variety of budgetary reallocations, such as increasing small subsidies while decreasing large ones.
\end{remark}

\subsection{Unobserved Heterogeneity and Endogenous Selection}
\label{sec:selection}

We now consider a setting where the baseline choice $W_i$ is not assigned by a policymaker but is instead an endogenous decision made by each agent. In this context, it is less natural to think of a change in the distribution of $W_i$ as a directly implementable reform. Nevertheless, it remains economically valuable to quantify how aggregate welfare responds to shifts in the distribution of these choices. To conduct this analysis, we assume the presence of exogenous variation in the form of an instrumental variable (IV), denoted by $Z_i$. We focus on a binary choice, $W_i \in \{0,1\}$, to simplify the exposition.
\begin{assumption}[Instrumental Variable]
\label{ass:iv}
There exists an instrument $Z_i$ for the binary choice $W_i$ that satisfies:
\begin{enumerate}
    \item \textbf{Random Assignment:} $Z_i$ is independent of all potential outcomes and reports.
    \item \textbf{Selection Model:} Selection is governed by the latent variable model $W_i = \mathbf{1}\{p(Z_i) > \xi_i\}$, where $\xi_i$ is uniform on $[0,1]$ and independent of $Z_i$.
    \item \textbf{Exclusion Restriction:} The instrument $Z_i$ does not directly enter the allocation mechanism, potential outcomes $Y_i(w,a)$, or potential reports $R_i(w)$.
\end{enumerate}
\end{assumption}
This setup describes a conventional selection model in the spirit of \citet{heckman1979sample}.\footnote{As shown by \citet{vytlacil2002independence}, this latent variable formulation is equivalent to the monotonicity assumption in the LATE framework of \citet{imbens1994identification}.} It opens two distinct avenues for policy analysis. The first is to treat the instrument $Z_i$ itself as the policy lever. Since $Z_i$ is randomly assigned, this case reduces to a direct application of our main result. In the context of Corollary~\ref{cor:psi}, the MPE with respect to a reform of the instrument's distribution, characterized by a score $s_Z$, is given by:
\begin{equation*}
    \mathcal{U}^{\prime}(s_{Z}) = \mathbb{E}_0[\Psi_i^{\text{total}} s_{Z}(Z_i)].
\end{equation*}
This is effectively an intention-to-treat (ITT) analysis using our equilibrium-adjusted outcome $\Psi_i^{\text{total}}$. Critically, because the market mechanism does not respond directly to $Z_i$, this ITT-type result does not require the exclusion restriction (part 3 of Assumption~\ref{ass:iv}).

The second, more structural avenue uses the MTE framework to analyze policies that target the endogenous choice $W_i$ \citep{bjorklund1987estimation,heckman2001policy,heckman2005structural}. For a binary instrument $Z_i \in \{0,1\}$, it is natural to consider the Wald-type estimand for our welfare-relevant outcome:
\begin{equation*}
    \frac{\mathbb{E}_0[\Psi_i^{\text{total}} | Z_i=1] - \mathbb{E}_0[\Psi_i^{\text{total}} | Z_i=0]}{\mathbb{E}_0[W_i | Z_i=1] - \mathbb{E}_0[W_i | Z_i=0]} = \frac{\text{Cov}(\Psi_i^{\text{total}}, Z_i)}{\text{Cov}(W_i, Z_i)}.
\end{equation*}
A key result from the MTE literature is that this ratio identifies the average treatment effect for the subpopulation of ``compliers''---those induced to change their choice by the instrument \citep{imbens1994identification}. By applying this logic to our structural outcome $\Psi_i^{\text{total}}$, we can show that this estimand is equal to the average MTE for the complier population:
\begin{equation*}
    \mathbb{E}[MTE_{\Psi^{\text{total}}}(\xi_i)| p(0) \le \xi_i \le p(1)], \quad \text{where} \quad MTE_{\Psi^{\text{total}}}(\xi) := \mathbb{E}[\Psi_i^{\text{total}}(1) - \Psi_i^{\text{total}}(0) | \xi_i = \xi].
\end{equation*}
This representation is powerful because it establishes a direct link between an estimable quantity (the Wald ratio for $\Psi_i^{\text{total}}$) and the MPE for a specific, economically meaningful policy. The policy is one that induces a uniform shift in the choice probability for the complier group, characterized by the score:
\begin{equation*}
    s_{W|\xi}(W_i|\xi_i) \propto \left(\frac{W_i}{\Pr(W_i=1|\xi_i)} - \frac{1-W_i}{\Pr(W_i=0|\xi_i)}\right) \cdot \mathbf{1}\{p(0) \le \xi_i \le p(1)\}.
\end{equation*}
This particular policy can be implemented by manipulating the distribution of the instrument $Z_i$. Whether such a manipulation is a practical policy or a purely theoretical benchmark depends on the nature of $Z_i$ itself.

This logic extends directly to a discrete instrument with $L+1$ support points, $Z_i \in \{0, \dots, L\}$. In this case, we can identify the MPE for any policy that targets a linear combination of the $L$ complier groups (those with $\xi_i \in [p(l), p(l+1)]$). A policymaker can then choose the weights on these groups to find the optimal implementable local reform. To evaluate a broader class of policies---those that cannot be implemented by simply re-weighting the instrument---one must rely on additional assumptions to extrapolate the MTE curve beyond the identified regions. The literature provides extensive tools for such exercises, from parametric assumptions \citep{brinch2017beyond} to partial identification approaches \citep{mogstad2018using}. In the latter case, the question of the optimal local reform is directly connected to a robust policy design under ambiguity \citep{manski2011choosing}.

\begin{remark}[Non-Binary Choices]
Our focus on a binary choice $W_i$ is for expositional simplicity. The core logic developed here extends to settings with non-binary choices, such as discrete or ordered choice models and non-ordered instruments. Although the MTE framework is most developed for the binary case, a growing literature provides the necessary tools for these richer settings. For a recent and comprehensive overview, see the Handbook chapter by \citet{mogstad2024instrumental}.
\end{remark}

\subsection{Discussion: Strategic Reporting in Non-Strategy-Proof Mechanisms}
\label{sec:strategic}

Our analysis has so far assumed that an individual's report to the mechanism is a stable function of the policy instrument. This is reasonable in strategy-proof environments, but many real-world markets are not. In such settings, rational individuals adapt their reports to the market environment. A policy that alters this environment will therefore induce a strategic response, adding a new channel through which welfare is affected.

\paragraph{A Framework for Strategic Reporting.} To analyze these settings, we distinguish between an agent's latent ``true'' type, $R_i^\star = R_i^\star(W_i)$, and their strategically chosen report, $R_i$. The path to identification depends critically on the informational content of the observed reports. In some environments, such as a first-price auction, economic theory provides an invertible mapping from reports to types, allowing $R_i^\star$ to be point-identified for each agent \citep{guerre2000optimal}. In such cases, our previous analysis applies directly to the recovered true types.

In more complex settings like matching markets, however, point-identification often fails. A potential way forward is to first recover the distribution of latent types in the baseline equilibrium, following methods like those in \citet{agarwal2018demand}. One can then model the strategic reporting strategy as a conditional distribution, $f_{R|R^{\star}}(R_i|R_i^{\star}, P_{R^{\star}})$, which captures how submitted reports respond to the competitive environment (summarized by the distribution of true types, $P_{R^{\star}}$). A marginal policy reform now propagates through two channels: its direct effect on the distribution of true types, $P_{R^{\star}}$, and its indirect effect on reporting strategies. If the strategic response is smooth, we can linearize it, leading to a total score that is the sum of the baseline policy score and a new strategic-response score. The MPE from Corollary~\ref{cor:psi} would then be:
\begin{equation*}
    \mathcal{U}^{\prime}(s_{W}) = \mathbb{E}_0[\Psi_{i}^{\text{total}}(s_W(W_i) + s_{R|R^{\star}}(R_i|R_i^{\star}, s_W))],
\end{equation*}
where $s_{R|R^{\star}}$ captures the strategic adjustment. Importantly, the structural object $\Psi_{i}^{\text{total}}$ remains invariant, as the market mechanism observes and responds only to the submitted reports $R_i$, not the latent types.

\paragraph{Identification Challenges.} While this representation provides a theoretical path forward, identifying the MPE in practice faces at least two significant challenges. The first challenge concerns the market conduct externality, $\Psi^{\text{conduct}}$. Its identification relies on the continuity of conditional mean outcomes at the allocation margin. As \citet{bertanha2023causal} show, this continuity can be violated in markets with strategic reporting. The intuition is that agents with knowledge of market cutoffs can strategically sort around them, invalidating the local comparisons underlying RDD-type estimands. \citet{bertanha2023causal} propose a solution using partial identification, deriving bounds on the marginal causal effects. These could, in principle, be used to derive bounds for the MPE.

The second, more fundamental challenge involves the correlation between an agent's outcome, $Y_i$, and their unobserved true type, $R_i^\star$. Even if the joint distribution of $(R_i, R_i^{\star})$ is identified, as in \citet{agarwal2018demand}, this is not sufficient to compute the expectation in the MPE formula. Doing so requires the joint distribution of $(Y_i, R_i, R_i^{\star})$. The situation is analogous to a standard selection model, where $R_i$ is an observed choice and the latent type $R_i^\star$ is an unobserved state that may be correlated with the outcome $Y_i$, conditional on the choice. Extending methods from related problems, such as in \citet{kline2016evaluating}, could offer a path forward, but would likely require non-trivial extrapolation and further assumptions. Thus, applying our framework in non-strategy-proof settings requires confronting deep identification problems, likely leading to partial identification of the MPE.
\section{Conclusion}\label{sec:conclusion}

Evaluating policies in centralized markets is complicated by equilibrium spillovers, which standard methods often fail to capture. The conventional wisdom is that identifying these effects requires observing the system's response to variation across different policy environments. This paper challenges that view by developing a framework to identify the total welfare effect of any marginal policy reform using data from within a single market. Our solution is the construction of the equilibrium-adjusted outcome ($\Psi_i^{\text{total}}$), a policy-invariant structural object that augments an agent's private outcome with the full equilibrium externalities they impose on others. A key insight of our approach is that the building blocks for this object's externality terms are often precisely the Local Average Treatment Effects (LATEs) identified in the empirical Regression Discontinuity Design (RDD) literature. This construction yields a practical "separation principle," where the marginal policy effect is a simple covariance between the policy score and $\Psi_i^{\text{total}}$.

This framework provides a toolkit with several applications. For policymakers, it offers a direct way to evaluate the "bang-for-the-buck" of the iterative, marginal policy changes that are common in practice. For researchers, it serves as a disciplined first step for analyzing global reforms. By sharply delineating what is identified non-parametrically from the data, our results provide a transparent foundation upon which any extrapolation required for global analysis must be built. More ambitious structural models can also be disciplined by requiring them to reproduce the identified local effects our framework provides. The framework's flexibility is further demonstrated by its extensions to optimal policy targeting and its novel connection to the Marginal Treatment Effects (MTE) literature for analyzing endogenous choice. By bridging reduced-form empirical work with the equilibrium structure of the market, our results offer a path toward more robust, data-driven policy design in a wide array of important economic settings.

\bibliography{./references}
\bibliographystyle{econ-aea}

\newpage
\appendix
\begin{center}
    \Large{Online Appendix}
\end{center}

\counterwithin{figure}{section}
\counterwithin{table}{section}
\counterwithin{equation}{section}

\footnotesize

\section{Differentiating integrals}\label{app:integrals}

Let $X$ be a random variable on $\mathbb{R}^k$ with a probability density function $p(x)$. Let $I \subset \mathbb{R}^p$ be an open set of parameters. We are interested in the differentiability of the functional $\mathcal{U}: I \to \mathbb{R}$ defined as:
\begin{align*}
\mathcal{U}(\mathbf{c}) = \E[h(X, \mathbf{c}) \1\{\phi(X, \mathbf{c}) \ge 0\}] = \int_{\mathbb{R}^k} h(x, \mathbf{c}) p(x) \1\{\phi(x, \mathbf{c}) \ge 0\} \,dx
\end{align*}
This functional is central to our analysis, representing the aggregate welfare where $h(x, \mathbf{c})$ is an agent's outcome and the indicator function reflects an eligibility rule determined by market-clearing parameters $\mathbf{c}$.

\begin{assumption} \label{ass:main_appendix_A}
The functions $h(x, \mathbf{c})$, $p(x)$, and $\phi(x, \mathbf{c})$ satisfy the following conditions:
\begin{enumerate}
    \item \textbf{Integrability:} The density $p(x)$ is bounded and continuous; the function $x \mapsto h(x, \mathbf{c}_0)$ is integrable.
    \item \textbf{Differentiability:} For almost every $x \in \mathbb{R}^k$ function $\mathbf{c} \mapsto h(x,\mathbf{c})$ is differentiable in the neighbourhood of $\mathbf{c}_0$ and its derivative is uniformly bounded by some integrable function $K_h(x)$.
    \item \textbf{Continuity:} Function $h(x, \mathbf{c}_0)p(x)$ is continuous in the open neighbourhood of $\{x | \phi(x, \mathbf{c}_0) = 0\}$.
    \item \textbf{Boundary Non-degeneracy:} Function $(x, \mathbf{c}) \mapsto \phi(x, \mathbf{c})$ is Lipschitz continuous and $\|\nabla_x \phi(x, \mathbf{c}_0)\|_2 > \epsilon >0$ $\mathcal{H}^{k-1}$-a.s. on the boundary surface $\{\phi(x,\mathbf{c}_0) = 0\}$.
\end{enumerate}
\end{assumption}

\begin{theorem}\label{th:main_tech}
Under Assumption \ref{ass:main_appendix_A}, the functional $\mathcal{U}(\mathbf{c})$ is differentiable at $\mathbf{c}_0 \in I$. Its partial derivative with respect to $c_j$ is given by:
\begin{align*}
\frac{\partial \mathcal{U}}{\partial c_j}(\mathbf{c}_0) = 
\int_{\mathbb{R}^k} \frac{\partial h(x, \mathbf{c}_0)}{\partial c_j} p(x) \1\{\phi(x, \mathbf{c}_0) \ge 0\} \,dx
+
\int_{\{x \,|\, \phi(x,\mathbf{c}_0)=0\}} h(x, \mathbf{c}_0)p(x) \frac{\partial \phi(x, \mathbf{c}_0) / \partial c_j}{\|\nabla_{x} \phi(x, \mathbf{c}_0)\|_2} d \mathcal{H}^{k-1}(x)
\end{align*}
where $\mathcal{H}^{k-1}$ is the $(k-1)$-dimensional Hausdorff measure.
\end{theorem}

\begin{proof}
The proof proceeds by analyzing the limit of the difference quotient for $\mathcal{U}(\mathbf{c})$ at $\mathbf{c}_0$. Let $\mathbf{e}_j$ be the $j$-th standard basis vector in $\mathbb{R}^p$. The partial derivative is the limit:
\begin{align*}
\frac{\partial \mathcal{U}}{\partial c_j}(\mathbf{c}_0) = \lim_{t \to 0} \frac{\mathcal{U}(\mathbf{c}_0 + t\mathbf{e}_j) - \mathcal{U}(\mathbf{c}_0)}{t}
\end{align*}
Let $g(x, \mathbf{c}) = h(x, \mathbf{c}) p(x)$ and let $H(z) = \1\{z \ge 0\}$ be the Heaviside step function. The difference quotient can be written as:
\begin{align*}
\frac{1}{t} \int_{\mathbb{R}^k} \left[ g(x, \mathbf{c}_0 + t\mathbf{e}_j) H(\phi(x, \mathbf{c}_0 + t\mathbf{e}_j)) - g(x, \mathbf{c}_0) H(\phi(x, \mathbf{c}_0)) \right] \,dx
\end{align*}
We add and subtract $g(x, \mathbf{c}_0) H(\phi(x, \mathbf{c}_0 + t\mathbf{e}_j))$ inside the brackets to separate the expression into two parts:
\begin{align*}
    \text{Term 1: } & \int_{\mathbb{R}^k} \frac{g(x, \mathbf{c}_0 + t\mathbf{e}_j) - g(x, \mathbf{c}_0)}{t} H(\phi(x, \mathbf{c}_0 + t\mathbf{e}_j)) \,dx \\
    \text{Term 2: } & \int_{\mathbb{R}^k} g(x, \mathbf{c}_0) \frac{H(\phi(x, \mathbf{c}_0 + t\mathbf{e}_j)) - H(\phi(x, \mathbf{c}_0))}{t} \,dx
\end{align*}
For the first term, by the Mean Value Theorem, the integrand is equal to $\frac{\partial g(x, \tilde{\mathbf{c}})}{\partial c_j}$ for some $\tilde{\mathbf{c}}$ on the line segment between $\mathbf{c}_0$ and $\mathbf{c}_0 + t\mathbf{e}_j$. By Assumption~\ref{ass:main_appendix_A}(2), this is bounded in absolute value by the integrable function $K_h(x)p(x)$. Therefore, the Dominated Convergence Theorem applies, and Term 1 converges to:
\begin{align*}
    \int_{\mathbb{R}^k} \frac{\partial h(x, \mathbf{c}_0)}{\partial c_j} p(x) \1\{\phi(x, \mathbf{c}_0) \ge 0\} \,dx
\end{align*}
The second term is the derivative of a function $\mathbf{c} \mapsto \int_{\{x|\phi(x, \mathbf{c}) \ge 0\}} h(x, \mathbf{c}_0)p(x) dx$. We appeal to the theory of shape derivatives to evaluate this. The conditions in Assumption~\ref{ass:main_appendix_A} are sufficient to apply Theorem 4.2 in \citet{delfour2011shapes}, which guarantees that the derivative exists and is equal to: 
\begin{align*}
    \int_{\{x \,|\, \phi(x,\mathbf{c}_0)=0\}} h(x, \mathbf{c}_0)p(x) \frac{\partial \phi(x, \mathbf{c}_0) / \partial c_j}{\|\nabla_{x} \phi(x, \mathbf{c}_0)\|_2} d \mathcal{H}^{k-1}(x),
\end{align*}
%
%
thus proving the result.
\end{proof}

\begin{corollary}\label{cor:cond_diff}
Let $X = (Y, Z)$, where $Y \in \R^{k_1}, Z \in \R^{k_2}$. Consider $\mathcal{U}(\mathbf{c}) = \E[h(Y,Z,\mathbf{c}) \1\{\phi(Y,Z, \mathbf{c}) \ge 0\}]$. Suppose for $P_Z$-almost every $z$, the assumptions of Theorem~\ref{th:main_tech} hold for the conditional expectation over $Y$. Let $g_j(\mathbf{c}|z)$ be the resulting conditional derivative. If $|g_j(\mathbf{c}|z)|$ is dominated by an integrable function $K(z)$, then $\mathcal{U}(\mathbf{c})$ is differentiable and
\begin{align*}
\frac{\partial \mathcal{U}(\mathbf{c})}{\partial \mathbf{c}_j} = \E_Z[g_j(\mathbf{c}|Z)].
\end{align*}
\end{corollary}
\begin{proof}
By the law of total expectation, $\mathcal{U}(\mathbf{c}) = \E_Z[\mathcal{U}(\mathbf{c}|Z)]$. For a.e. $z$, $\mathcal{U}(\mathbf{c}|z)$ is differentiable with derivative $g_j(\mathbf{c}|z)$ by Theorem~\ref{th:main_tech}. The domination condition allows us to apply the Differentiated DCT to interchange the derivative and the outer expectation $\E_Z[\cdot]$.
\end{proof}
\newpage

\section{Derivation of the Marginal Policy Effect}\label{app:mpe_proof}

This appendix provides a formal derivation of the Marginal Policy Effect (MPE) discussed in Section~\ref{sec:derivation} of the main text. We first state the rigorous versions of the assumptions used in the derivation. We then prove a sequence of lemmas and propositions that decompose the MPE into its constituent parts: a direct effect, a competition effect, and a market conduct effect. The proof relies on the result for differentiating integrals over moving domains from Appendix~\ref{app:integrals}.

Let the space of observable data for an agent be $\mathcal{D} := \mathcal{Y}\times \mathcal{A}\times \mathcal{R}\times \mathcal{W}$. We assume the baseline policy regime is characterized by a probability measure on $\mathcal{D}$ that has a density with respect to a product measure $\lambda$. This density can be factorized as:
\begin{align*}
    f_{Y|A,W,R}(y|a,w,r) \mu_{a}(r, \mathbf{c}_0, P_{R|0}) f_{R|W}(r|w) f_{W|0}(w),
\end{align*}
where $\mathbf{c}_0$ is the equilibrium parameter vector and $P_{R|0}$ is the marginal measure on the report space $\mathcal{R}$, induced by the baseline policy $f_{W|0}$. Its density is $f_{R|0}(r) = \int f_{R|W}(r|w) f_{W|0}(w) dw$. We use $\E_0$ to denote the expectation with respect to this baseline measure.

A marginal policy reform is a perturbation of the baseline policy distribution $f_{W|0}$ in a specific direction. We characterize these directions by a set of score functions.
\begin{definition}[Policy Score Space]
The space of admissible policy scores, $\mathcal{S}_W$, is the set of functions $s_W: \mathcal{W} \to \R$ such that $\E_0[s_W(W_i)] = 0$ and $\E_0[s^2_W(W_i)] < \infty$.
\end{definition}
For any given score $s_W \in \mathcal{S}_W$, we can construct a local path of policy distributions indexed by a parameter $\theta \in \R$. A standard construction that accommodates all scores in $\mathcal{S}_W$ is the linear path:
\begin{align*}
    f_W(w|\theta, s_W) = f_{W|0}(w) (1 + \theta s_W(w))
\end{align*}
This path is well-defined for $\theta$ in a neighborhood of 0. It satisfies $f_W(w|0, s_W) = f_{W|0}(w)$ and guarantees that the score of the log-likelihood with respect to $\theta$ at $\theta=0$ is precisely $s_W(w)$:
\begin{align*}
    \frac{\partial}{\partial \theta} \log f_W(w|\theta, s_W) \bigg|_{\theta=0} = \frac{f_{W|0}(w)s_W(w)}{f_{W|0}(w)} = s_W(w).
\end{align*}
A reform to the policy distribution $f_W$ induces a change in the marginal distribution of reports $f_R$. The perturbed report density is given by:
\begin{align*}
    f_R(r|\theta, s_W) = \int f_{R|W}(r|w) f_W(w|\theta, s_W) dw.
\end{align*}
This perturbation of the report distribution also has a well-defined score, $s_R(r)$, which is characterized by the following result.
\begin{lemma}[Induced Score]\label{lemma:ind_score}
The score of the induced report distribution, $s_R(r) = \frac{\partial}{\partial\theta} \log f_R(r|\theta,s_W)|_{\theta=0}$, is the conditional expectation of the policy score:
\begin{align*}
    s_R(r) = \E_0[s_W(W_i) | R_i=r].
\end{align*}
We use $\mathcal{S}_{R}$ to denote the set of scores $s_{R}(R_i)$ induced by $\mathcal{S}_{W}$.
\end{lemma}
\begin{proof}
By definition, $s_R(r) = f_R'(r|0) / f_{R|0}(r)$, where the prime denotes the partial derivative with respect to $\theta$ at $\theta=0$. We have:
\begin{align*}
    f_R'(r|0) &= \int f_{R|W}(r|w) f_W'(w|0) dw \\
    &= \int f_{R|W}(r|w) [s_W(w) f_{W|0}(w)] dw \\
    &= \int \frac{f_{R|W}(r|w) f_{W|0}(w)}{f_{R|0}(r)} s_W(w) f_{R|0}(r) dw \\
    &= f_{R|0}(r) \int f_{W|R}(w|r) s_W(w) dw = f_{R|0}(r) \E_0[s_W(W_i)|R_i=r].
\end{align*}
Here the differentiation under the integral is justified because $f_{W}(w|\theta, s_W)$ is linear in $\theta$ and the DCT can be applied because $\E_0[s^2_{W}(W_i)] < \infty$. Dividing by $f_{R|0}(r)$ yields the result.
\end{proof}

\begin{assumption}[Marginal Agent Regularity] \label{ass:marginal_agents_formal}
For each $a \in \mathcal{A}$ define $m_{a}(r) := \mathbb{E}_0[Y_i(W_i, a)|R_i =r]$; we assume
\begin{enumerate}
    \item \textbf{Report Structure:} The report vector can be decomposed as $R_i = (R_{i,un}, R_{i,cont})$, where, conditional on $R_{i,un}$, the distribution of $R_{i,cont}$ is absolutely continuous with respect to the Lebesgue measure on $\mathbb{R}^{k}$ with a bounded and continuous density function.
    
    \item \textbf{Continuity of Conditional Outcomes:} For each allocation $a \in \mathcal{A}$ and fixed $r_{un}$, the conditional mean function $r_{cont} \mapsto m_a(r_{un}, r_{cont})$ is continuous and bounded.
\end{enumerate}
\end{assumption}

\begin{assumption}[Well-Behaved Mechanism] \label{ass:mechanism_formal}
For any $a \in \mathcal{A}$ the allocation probability, $\mu_a(r, \mathbf{c}, P_R)$, can be decomposed as:
\begin{align*}
    \mu_a(r, \mathbf{c}, P_R) = h_a(r, \mathbf{c}, P_R) \cdot \ind\{\phi_a(r, \mathbf{c}) \ge 0\}
\end{align*}
where the components satisfy the following conditions at the baseline equilibrium $(\mathbf{c}_0, P_{R|0})$:
\begin{enumerate}
    \item $P_{R_un|0}$-a.s. the functions $(r_{cont}, \mathbf{c}) \mapsto \phi_a(r_{un},r_{cont}, \mathbf{c})$,  $(r_{cont}, \mathbf{c}) \mapsto h_a(r_{un},r_{cont}, \mathbf{c},P_{R|0})$ satisfy the conditions laid out in Assumption~\ref{ass:main_appendix_A} of Appendix~\ref{app:integrals}.
    
    \item  $P_{R|0}$-a.s. the function $(\mathbf{c}, P_R) \mapsto h_a(r, \mathbf{c}, P_R)$ is Hadamard differentiable with respect to $P_R$ at $P_{R|0}$ for paths induced by scores $s_R \in \mathcal{S}_R$. Its derivative in the direction $s_R$ is a continuous linear functional given by:
    \begin{align*}
        D_P h_a(r, \mathbf{c}, P_{R|0})[s_R] = \int L_{a,\mathbf{c}}(r, r') s_R(r') dr'
    \end{align*}
    where $L_{a,\mathbf{c}}(r, r')$ is the square-integrable kernel and the mapping $\mathbf{c} \mapsto L_{a,\mathbf{c}}(r, r')$ is continuous at $\mathbf{c}_0$.
\end{enumerate}
\end{assumption}

\begin{assumption}[Differentiability of the Market Conduct Rule] \label{ass:conduct_rule_formal}
The market conduct rule $\mathbf{c}(P_R)$ is Hadamard differentiable at the baseline report distribution $P_{R|0}$ along the paths induced by scores $s_R$. Its derivative, a continuous linear operator from the space of scores to $\R^{p}$, has the representation:
\begin{align*}
    \mathbf{c}'[s_R] = \langle \boldsymbol{\psi}_{c_0}, s_R \rangle_{\mathcal{H}_R}
\end{align*}
where $\langle \cdot, \cdot \rangle_{\mathcal{H}_R}$ is the inner product on a Hilbert space $\mathcal{H}_R$ containing the scores, and $\boldsymbol{\psi}_{c_0}$ is the representer of the derivative.
\end{assumption}

We will split the proof of Theorem~\ref{th:mpe_general} and Corollary~\ref{cor:psi} into Lemmas~\ref{lemma:direct}-\ref{lemma:id}. The final result follows from the direct combination of the intermediate results. 

\begin{lemma}[The Direct Effect]\label{lemma:direct} Suppose Assumption~\ref{ass:marginal_agents_formal} holds.
The direct effect of a policy reform, defined as the welfare impact of perturbing the policy distribution while holding the allocation rule fixed at the baseline $(\mathbf{c}_0, P_{R|0})$, is given by $\E_0[Y_i s_W(W_i)]$.
\end{lemma}
\begin{proof}
The welfare functional for the direct effect is 
\begin{align*}
    \mathcal{U}_{\text{direct}}(\theta):= \int y f_{Y|A,W,R} \mu_a(r, \mathbf{c}_0, P_{R|0}) f_{R|W} f_W(w|\theta,s_W) d\lambda.
\end{align*}
Since $\mu_a$ is held fixed, this is a standard expectation. The score of the density of the data with respect to $\theta$ is simply $s_W(w)$. The derivative of the expectation is the expectation of the outcome multiplied by the score, which gives $\E_0[Y_i s_W(W_i)]$. The differentiation under the integral is permitted by Assumption~\ref{ass:marginal_agents_formal} which guarantees that the conditional expectation of $Y_i$ is bounded. 
\end{proof}

Next, we consider the indirect effect. To this end, we define 
\begin{align*}
    \mathcal{U}_{\text{part}}(\theta) := \sum_{a= 0}^K\int m_{a}(r) \mu_{a}(r, \mathbf{c}(\theta, s_W), P_{R|\theta, s_W}) f_{R|0}(r)dr
\end{align*}
We define this auxiliary functional, which evolves the market mechanism but holds the population of agents fixed to the baseline distribution, as a tool to isolate the indirect effects via the chain rule.

\begin{lemma}[Decomposition of the Indirect Effect]\label{lemma:decomp} Suppose Assumptions~\ref{ass:marginal_agents_formal}-\ref{ass:conduct_rule_formal} hold. Then $\mathcal{U}^{\prime}_{\text{part}}(0)$ exists and is equal to:
\begin{align*}
\mathcal{U}'_{\text{part}}(0) =
\sum_{a=0}^K\mathbb{E}_0[m_{a}(R_i) L_{a,\mathbf{c}_0}(R_i, R_j) s_R(R_j)] + \langle \nabla_c\mathcal{U}(c_0, P_{R|0}) \cdot \boldsymbol{\psi}_{c_0}, s_R \rangle_{\mathcal{H}_R},
\end{align*}
where 
\begin{align*}
    \nabla_c\mathcal{U}(c_0, P_{R|0}) = &
    \sum_{a = 0}^{K}
    \mathbb{E}_0\left[m_{a}(R_i)\nabla_{\mathbf{c}} h_a(R_i, \mathbf{c}_0, P_{R|0}) \1\{\phi_a(R_i, \mathbf{c}_0) \ge 0\}\right] + \\
    & \sum_{a = 0}^{K}\E_0\left[\mathbb{E}_0\left[ m_a(R_i)h_a(R_i, \mathbf{c}_0,P_{R|0}) \nabla_{\mathbf{c}} \phi_a(R_i, \mathbf{c}_0)\mid R_{i,un}, \phi_a(R_i, \mathbf{c}_0) = 0 \right] f_{\phi_a}(0\mid R_{i,un})\right],
\end{align*}
and $f_{\phi_a}(0\mid r_{un})$ is the conditional density of $\phi_a(R_i, \mathbf{c}_0)$ given $R_{i,un} = r_{un}$.
\end{lemma}
\begin{proof}
The derivative of $\mathcal{U}_{\text{part}}(\theta)$ at $\theta=0$ can be found by applying the chain rule to the underlying functional $\mathcal{U}(\mathbf{c}, P_R) = \sum_a \int m_a(r) \mu_a(r, \mathbf{c}, P_R) f_{R|0}(r) dr$. The derivative is the sum of the partial derivatives with respect to each argument, evaluated along the path of the reform:
\begin{align*}
    \mathcal{U}'_{\text{part}}(0) = \underbrace{D_P \mathcal{U}(\mathbf{c}_0, P_{R|0})[s_R]}_{\text{Competition Effect}} + \underbrace{\nabla_c \mathcal{U}(\mathbf{c}_0, P_{R|0}) \cdot \mathbf{c}'[s_R]}_{\text{Market Conduct Effect}}.
\end{align*}
The first term, the Competition Effect, is the functional derivative with respect to $P_R$ in the direction of the induced score $s_R$. Using Assumption~\ref{ass:mechanism_formal}, it is:
\begin{align*}
    D_P \mathcal{U}(\mathbf{c}_0, P_{R|0})[s_R] &= \sum_a \mathbb{E}_0[m_a(R_i) D_P \mu_a(R_i, \mathbf{c}_0, P_{R|0})[s_R]] \\
    &= \sum_a \mathbb{E}_0[m_a(R_i) \ind\{\phi_a(R_i, \mathbf{c}_0) \ge 0\} D_P h_a(R_i, \mathbf{c}_0, P_{R|0})[s_R]] \\
    &= \sum_a \mathbb{E}_0\left[m_a(R_i) \int L_{a, \mathbf{c}_0}(R_i, r') s_R(r') dr'\right] \\
    &= \sum_a \mathbb{E}_0[m_{a}(R_i) L_{a,\mathbf{c}_0}(R_i, R_j) s_R(R_j)].
\end{align*}
The second term is the Market Conduct Effect. The gradient $\nabla_c \mathcal{U}(\mathbf{c}_0, P_{R|0})$ is derived by applying Corollary~\ref{cor:cond_diff} from Appendix~\ref{app:integrals}, yielding the expression in the lemma statement. Combining this with the derivative of the market conduct rule from Assumption~\ref{ass:conduct_rule_formal}, $\mathbf{c}'[s_R] = \langle \boldsymbol{\psi}_{c_0}, s_R \rangle_{\mathcal{H}_R}$, gives the second part of the result. Summing the two effects proves the lemma.
\end{proof}

\begin{lemma}[Decomposition of the MPE]
Suppose Assumptions~\ref{ass:marginal_agents_formal}-\ref{ass:conduct_rule_formal} hold; then the total MPE can be decomposed into a direct effect and an indirect equilibrium effect:
\begin{align*}
    \mathcal{U}'(0) = \E_0[Y_i s_W(W_i)] + \mathcal{U}'_{\text{part}}(0).
\end{align*}
\end{lemma}
\begin{proof}
We analyze the difference quotient for $\mathcal{U}(\theta)$. Let $\mu_a(\theta) = \mu_a(r, \mathbf{c}(\theta), P_{R|\theta})$ and $\mu_a(0) = \mu_a(r, \mathbf{c}_0, P_{R|0})$. The difference quotient is:
\begin{align*}
    \frac{\mathcal{U}(\theta) - \mathcal{U}(0)}{\theta} = \frac{1}{\theta} \int y f_{Y|A,W,R} \left[ \mu_a(\theta) f_{R|W} f_W(\theta) - \mu_a(0) f_{R|W} f_W(0) \right] d\lambda
\end{align*}
We decompose the term in the brackets by adding and subtracting $\mu_a(0) f_{R|W} f_W(\theta)$:
\begin{align*}
    \frac{1}{\theta} \int y f_{Y|A,W,R} \left[ (\mu_a(\theta) - \mu_a(0)) f_{R|W} f_W(\theta) + \mu_a(0) f_{R|W} (f_W(\theta) - f_W(0)) \right] d\lambda
\end{align*}
The second term in this sum corresponds to the direct effect. As $\theta \to 0$, its limit is $\E_0[Y_i s_W(W_i)]$ by Lemma~\ref{lemma:direct}. We now focus on the first term, which captures the indirect effect:
\begin{align*}
    \text{Indirect Term} = \frac{1}{\theta} \int y f_{Y|A,W,R} (\mu_a(\theta) - \mu_a(0)) f_{R|W} f_W(\theta) d\lambda
\end{align*}
We further decompose this by writing $f_W(\theta) = f_W(0) + (f_W(\theta) - f_W(0))$:
\begin{align*}
    \text{Indirect Term} &= \frac{1}{\theta} \int y f_{Y|A,W,R} (\mu_a(\theta) - \mu_a(0)) f_{R|W} f_W(0) d\lambda \quad &(A) \\
    &+ \int y f_{Y|A,W,R} (\mu_a(\theta) - \mu_a(0)) f_{R|W} \frac{f_W(\theta) - f_W(0)}{\theta} d\lambda \quad &(B)
\end{align*}
The second part, Term (B), is negligible. As $\theta \to 0$, the factor $(\mu_a(\theta) - \mu_a(0))$ converges to 0 due to the continuity of the mechanism and market conduct rule. The other factor, $\frac{f_W(\theta) - f_W(0)}{\theta}$, converges to $s_W(w)f_{W|0}(w)$. Thus, the entire integrand converges pointwise to 0. The Dominated Convergence Theorem (justified by our assumptions) implies the integral of this term converges to 0.

The first part, Term (A), is the difference quotient for $\mathcal{U}_{\text{part}}(\theta)$. By integrating over $y$ and $w$ with the baseline policy $f_W(0)$, the expression becomes:
\begin{align*}
    \text{Term (A)} = \frac{1}{\theta} \left( \int m_a(r)\mu_a(\theta)f_{R|0}(r)dr - \int m_a(r)\mu_a(0)f_{R|0}(r)dr \right) = \frac{\mathcal{U}_{\text{part}}(\theta) - \mathcal{U}_{\text{part}}(0)}{\theta}
\end{align*}
Taking the limit as $\theta \to 0$ for all terms, we find that $\mathcal{U}'(0)$ is the sum of the direct effect and $\mathcal{U}'_{\text{part}}(0)$, which proves the result.
\end{proof}

\begin{lemma}[Identification]\label{lemma:id}
    Suppose Assumptions~\ref{ass:marginal_agents_formal}-\ref{ass:conduct_rule_formal} hold. Also, suppose $\E_0\left[\frac{1}{h_a(R_i, \mathbf{c}_0,P_{R|0})}\right] < \infty$ and $\E_0\left[\left(\frac{\|\nabla_{\mathbf{c}} h(R_i, \mathbf{c}_0, P_{R|0}) \|_2}{h_a(R_i, \mathbf{c}_0,P_{R|0})}\right)^2\right] <\infty$ Then the MPE is identified
\end{lemma}
\begin{proof}
To show identification, we must express the components of the MPE as expectations over the observed data distribution. The core of the argument is an inverse weighting identity that allows us to replace the unobserved $m_a(R_i)$ with the observed outcome $Y_i$. For any sufficiently regular function $g(a, r, r')$, the following identity holds:
\begin{align*}
    \sum_{a=0}^K \mathbb{E}_0[m_a(R_i) g(a, R_i, R_j) \mu_a(R_i, \mathbf{c}_0, P_{R|0})] = \mathbb{E}_0[Y_i g(A_i, R_i, R_j)].
\end{align*}
To see this, we apply the law of iterated expectations to the right-hand side, conditioning first on $(A_i, R_i, R_j)$:
\begin{align*}
    \mathbb{E}_0[Y_i g(A_i, R_i, R_j)] &= \mathbb{E}_0[\mathbb{E}_0[Y_i | A_i, R_i, R_j] g(A_i, R_i, R_j)] \\
    &= \mathbb{E}_0[m_{A_i}(R_i) g(A_i, R_i, R_j)] \\
    &= \sum_{a=0}^K \mathbb{E}_0[m_a(R_i) g(a, R_i, R_j) \ind\{A_i=a\}] \\
    &= \sum_{a=0}^K \mathbb{E}_0[m_a(R_i) g(a, R_i, R_j) \mathbb{P}(A_i=a | R_i, R_j)],
\end{align*}
where the second line follows from the definition of $m_a$ and the fact that $Y_i$ is independent of $R_j$ conditional on $(A_i, R_i)$. Since $\mathbb{P}(A_i=a | R_i, R_j) = \mu_a(R_i, \mathbf{c}_0, P_{R|0})$, the identity is established. The moment conditions in the Lemma statement ensure these expectations are well-defined.

We now apply this identity to the two terms.

\textbf{1. Competition Effect Term:}
Let $g(a, R_i, R_j) = \frac{L_{a,\mathbf{c}_0}(R_i, R_j)}{h_{a}(R_i, \mathbf{c}_0, P_{R|0})} s_R(R_j)$. The term is:
\begin{align*}
    \sum_{a=0}^K\mathbb{E}_0[m_{a}(R_i) L_{a,\mathbf{c}_0}(R_i, R_j) s_R(R_j)].
\end{align*}
Note that $L_{a,\mathbf{c}_0}(R_i, R_j)$ must be zero if agent $i$ is ineligible ($\phi_a(R_i, \mathbf{c}_0)<0$), as their allocation cannot be affected by others. Thus, we can write the term as:
\begin{align*}
     \sum_{a=0}^K\mathbb{E}_0\left[m_{a}(R_i) \frac{L_{a,\mathbf{c}_0}(R_i, R_j)}{h_{a}(R_i, \mathbf{c}_0, P_{R|0})}\mu_{a}(R_i, \mathbf{c}_0, P_{R|0}) s_R(R_j)\right].
\end{align*}
Using our identity, this is equal to $\mathbb{E}_0\left[Y_i \frac{L_{A_i,\mathbf{c}_0}(R_i, R_j)}{h_{A_i}(R_i, \mathbf{c}_0, P_{R|0})} s_R(R_j)\right]$, which simplifies to $\mathbb{E}_0[\gamma(R_j) s_W(W_j)]$. The assumption $\E_0[1/h_a] < \infty$ ensures this expression is well-defined.

\textbf{2. Inframarginal Market Conduct Term:}
Let $g(a, R_i, R_j) = \frac{\nabla_{\mathbf{c}}h_{a}(R_i, \mathbf{c}_0, P_{R|0})}{h_{a}(R_i, \mathbf{c}_0, P_{R|0})}$. The term is:
\begin{align*}
    \sum_{a=0}^K \mathbb{E}_0[m_a(R_i) \nabla_{\mathbf{c}}h_a(R_i, \mathbf{c}_0, P_{R|0}) \ind\{\phi_a(R_i, \mathbf{c}_0) \ge 0\}].
\end{align*}
Replacing the indicator with $\mu_a/h_a$ and applying the identity yields the identified expression:
\begin{align*}
    \mathbb{E}_0\left[Y_i\frac{\nabla_{\mathbf{c}}h_{A_i}(R_i, \mathbf{c}_0, P_{R|0})}{h_{A_i}(R_i, \mathbf{c}_0, P_{R|0})}\right].
\end{align*}
The assumption $\E_0[(\|\nabla_{\mathbf{c}} h_a\|/h_a)^2] < \infty$ ensures this is well-defined.

Finally, the boundary term in $\nabla_c \mathcal{U}$ involves the conditional mean $m_a(r)$ evaluated at the boundary surface. Since $m_a(r)$ is assumed to be continuous in $r_{cont}$, its value at the boundary is the limit of its values on the interior, which are identified from the observed data. This standard RDD-style argument ensures that $m_a(r)$ is identified for any $r$ at the boundary. This, together with the fact that $h_{a}$ and $\phi_a$ are known functions, guarantees that the entire boundary term is identified.
\end{proof}

\newpage

\section{Examples}\label{app:examples}

This appendix applies the general framework developed in Appendix~\ref{app:mpe_proof} to the canonical examples presented in Section~\ref{sec:examples} of the main text. For each example, we follow a three-step process: (1) we formally specify the model's components, (2) we verify that the model satisfies the key assumptions required for our main theorem, and (3) we derive the specific form of the equilibrium-adjusted outcome, $\Psi_i^{\text{total}}$, by simplifying the general formula.

\subsection{Price-Based Allocation}

\subsubsection{Model Specification}
The market allocates a single product ($a \in \{0,1\}$) with fixed supply $q \in (0,1)$.
\begin{itemize}
    \item \textbf{Reports:} Agents submit a continuous valuation $R_i \in \R_+$. We assume the baseline distribution of reports $P_{R|0}$ admits a continuous and positive density, $f_R(r)$.
    \item \textbf{Allocation Rule:} An agent receives the good if their report exceeds a market-clearing price or cutoff, $c \in \R_+$. The allocation probability is thus:
    \begin{equation*}
     \mu_1(R_i, c) = \1\{R_i > c\}, \quad \mu_0(R_i, c) = \1\{R_i \le c\}.
    \end{equation*}
    \item \textbf{Market Conduct Rule:} The cutoff $c_0$ is set to satisfy the supply constraint, i.e., it is the $(1-q)$-quantile of the report distribution:
    \begin{equation*}
     \E_0[\mu_1(R_i, c_0)] = \int_{c_0}^\infty f_R(r)dr = q.
    \end{equation*}
\end{itemize}

\subsubsection{Verification of Assumptions}
We verify the key assumptions from Appendix~\ref{app:mpe_proof}.
\begin{itemize}
    \item \textbf{Assumption~\ref{ass:mechanism_formal} (Well-Behaved Mechanism):} The allocation rule fits the required decomposition. For $a=1$, we have $\phi_1(R_i, c) = R_i - c$ and $h_1(R_i, c, P_R) = 1$. For $a=0$, we have $\phi_0(R_i, c) = c - R_i$ and $h_0(R_i, c, P_R) = 1$.
    \begin{enumerate}
        \item The functions $h_a=1$ and $\phi_a$ are continuously differentiable in $c$.
        \item The eligibility function $\phi_a$ does not depend on $P_R$.
        \item The smooth component $h_a=1$ is constant and thus trivially Hadamard differentiable in $P_R$, with a derivative kernel $L_a(r, r') = 0$.
        \item All components are continuous.
    \end{enumerate}
    \item \textbf{Assumption~\ref{ass:marginal_agents_formal} (Marginal Agents):} The report $R_i$ is fully continuous. We assume the conditional mean outcomes $m_a(r)$ are continuous in $r$, as stated in the main text. The non-degeneracy condition $\|\nabla_{r}\phi_a\|_2 = |1| = 1 > 0$ holds.
    \item \textbf{Assumption~\ref{ass:conduct_rule_formal} (Differentiability of Market Conduct):} The market conduct rule is defined implicitly by $G(c, P_R) := \int_c^\infty dP_R(r) - q = 0$. This is differentiable in $L_2$. By the Implicit Function Theorem, its derivative is $c'[s_R] = -(\partial G/\partial c)^{-1} (\partial G/\partial P_R)[s_R]$. We have $\partial G/\partial c = -f_R(c)$ and $(\partial G/\partial P_R)[s_R] = \E_0[s_R(R_i)\1\{R_i > c\}]$. Thus, $c'[s_R]$ is a continuous linear functional in $L_2$, and the assumption holds.
\end{itemize}
Since the market conduct rule is differentiable in $L_2$, Corollary~\ref{cor:psi} applies.

\subsubsection{Derivation of $\Psi_i^{\text{total}}$}
We start with the general formula from Corollary~\ref{cor:psi}: $\Psi_i^{\text{total}} = \Psi_i^{\text{fixed}} + \Psi^{\text{conduct}}(R_i)$.

\begin{enumerate}
    \item \textbf{Fixed Component $\Psi_i^{\text{fixed}} = Y_i + \gamma(R_i)$:}
    The competition externality, $\gamma(R_i)$, depends on the kernel $L_a$. Since $h_a$ is constant with respect to $P_R$, its derivative kernel $L_a(r,r')$ is identically zero. Therefore, $\gamma(R_i) = 0$, and $\Psi_i^{\text{fixed}} = Y_i$.

    \item \textbf{Market Conduct Component $\Psi^{\text{conduct}}(R_i) = \nabla_c \mathcal{U}(c_0) \cdot \psi_{c_0}(R_i)$:}
    We need to derive the welfare gradient $\nabla_c \mathcal{U}(c_0)$ and the influence function of the cutoff, $\psi_{c_0}(R_i)$.
    \begin{itemize}
        \item \textbf{Welfare Gradient $\nabla_c \mathcal{U}(c_0)$:} We apply Theorem~\ref{th:main_tech}. Since $h_a=1$ does not depend on $c$, the volume term is zero. The derivative comes entirely from the boundary term. For $a=1$, the boundary is at $R_i=c$, and for $a=0$, it is also at $R_i=c$. A marginal increase in $c$ moves agents from allocation 1 to 0. The total effect is the mass of agents at the boundary, $f_R(c_0)$, times their change in average outcome, $m_0(c_0) - m_1(c_0)$.
        \begin{equation*}
         \nabla_c \mathcal{U}(c_0) = -[m_1(c_0) - m_0(c_0)]f_R(c_0) = -\tau(c_0)f_R(c_0).
        \end{equation*}
        \item \textbf{Influence Function $\psi_{c_0}(R_i)$:} From the verification of Assumption~\ref{ass:conduct_rule_formal}, we have $c'[s_R] = \E_0[s_R(R_i) \frac{\1\{R_i > c_0\}}{f_R(c_0)}]$. By definition, the influence function is the term inside the expectation multiplying the score:
        \begin{equation*}
         \psi_{c_0}(R_i) = \frac{\1\{R_i > c_0\}}{f_R(c_0)} = \frac{A_i}{f_R(c_0)}.
        \end{equation*}
    \end{itemize}
    Combining these pieces, the market conduct externality is:
    \begin{align*}
    \Psi^{\text{conduct}}(R_i) &= \nabla_c \mathcal{U}(c_0) \cdot \psi_{c_0}(R_i) \\
    &= \left( -\tau(c_0)f_R(c_0) \right) \cdot \left( \frac{A_i}{f_R(c_0)} \right) = -\tau(c_0)A_i.
    \end{align*}
    \item \textbf{Total Equilibrium-Adjusted Outcome:}
    Summing the components gives the final result:
    \begin{equation*}
     \Psi_i^{\text{total}} = Y_i - \tau(c_0)A_i.
    \end{equation*}
    This matches the expression in the main text. The density terms $f_R(c_0)$ cancel out, revealing that the externality an inframarginal agent imposes by taking a slot is exactly the welfare loss of the single marginal agent they displace.
\end{enumerate}

\subsection{School Choice with Multiple Schools}

\subsubsection{Model Specification}
We consider a market with two schools ($k \in \{1,2\}$) and an outside option ($a=0$). Each school has a fixed capacity $q_k$.
\begin{itemize}
    \item \textbf{Reports:} An agent's report is a vector $R_i = (\succ_i, V_{i,1}, V_{i,2})$, consisting of a strict preference ranking $\succ_i$ over $\{0,1,2\}$ and a vector of school-specific continuous scores (e.g., grades or test scores).
    \item \textbf{Allocation Rule:} The market is cleared by a vector of score cutoffs $\mathbf{c}=(c_1, c_2)$. An agent is eligible for school $k$ if their score exceeds the cutoff, $V_{i,k} > c_k$. Each agent is assigned to their most-preferred school for which they are eligible. If they are not eligible for any school they prefer to the outside option, they receive $a=0$.
    \item \textbf{Market Conduct Rule:} The cutoff vector $\mathbf{c}_0$ is determined by the system of capacity constraints:
    \begin{align*}
        \E_0[\mu_1(R_i, \mathbf{c}_0)] &= q_1 \\
        \E_0[\mu_2(R_i, \mathbf{c}_0)] &= q_2
    \end{align*}
\end{itemize}

\subsubsection{Verification of Assumptions}
\begin{itemize}
    \item \textbf{Assumption~\ref{ass:mechanism_formal} (Well-Behaved Mechanism):} For any given preference ranking $\succ_i$, the allocation rule is a series of indicator functions based on the scores. For example, for an agent with preferences $1 \succ_2 \succ 0$, the allocation rule is $\mu_1(R_i, \mathbf{c}) = \1\{V_{i,1} > c_1\}$ and $\mu_2(R_i, \mathbf{c}) = \1\{V_{i,1} \le c_1, V_{i,2} > c_2\}$. Each of these can be written in the required form with $h_a=1$ and $\phi_a$ being a function of the score margins (e.g., $\phi_1 = V_{i,1}-c_1$). The conditions on smoothness, anonymity, and Hadamard differentiability (with $L_a=0$) hold trivially.
    \item \textbf{Assumption~\ref{ass:marginal_agents_formal} (Marginal Agents):} The scores $(V_{i,1}, V_{i,2})$ are the continuous components of the report, while the preference ranking $\succ_i$ is the discrete component. We assume $m_a(R_i)$ is continuous in the scores. The non-degeneracy condition on the gradients of $\phi_a$ holds.
    \item \textbf{Assumption~\ref{ass:conduct_rule_formal} (Differentiability of Market Conduct):} The market conduct rule is defined by the system of equations $\mathbf{G}(\mathbf{c}, P_R) := \E_{P_R}[\boldsymbol{\mu}(R_i, \mathbf{c})] - \mathbf{q} = \mathbf{0}$. The Jacobian of this system is $J_{kj} = \partial G_k / \partial c_j$. We assume this matrix is invertible at $\mathbf{c}_0$, which requires that the cross-school substitution effects are not perfectly collinear. The rule is differentiable in $L_2$.
\end{itemize}
As in the previous example, Corollary~\ref{cor:psi} applies.

\subsubsection{Derivation of $\Psi_i^{\text{total}}$}
We again use the formula $\Psi_i^{\text{total}} = Y_i + \gamma(R_i) + \Psi^{\text{conduct}}(R_i)$.

\begin{enumerate}
    \item \textbf{Fixed Component $\Psi_i^{\text{fixed}}$:} Since the allocation rule (conditional on preferences) does not depend on $P_R$, the derivative kernel $L_a$ is zero. This implies the competition externality is $\gamma(R_i) = 0$, and thus $\Psi_i^{\text{fixed}} = Y_i$.

    \item \textbf{Market Conduct Component $\Psi^{\text{conduct}}(R_i)$:} This component is given by $\langle \nabla_{\mathbf{c}} \mathcal{U}(\mathbf{c}_0), \boldsymbol{\psi}_{\mathbf{c}_0}(R_i) \rangle$, where the gradient and influence function are now vectors.
    \begin{itemize}
        \item \textbf{Influence Function $\boldsymbol{\psi}_{\mathbf{c}_0}(R_i)$:} Applying the Implicit Function Theorem to the vector market conduct rule gives $\mathbf{c}'[s_R] = -J_0^{-1} (\partial \mathbf{G}/\partial P_R)[s_R]$. The second term is $\E_0[s_R(R_i)\boldsymbol{\mu}(R_i, \mathbf{c}_0)] = \E_0[s_R(R_i)\mathbf{A}_i]$, where $\mathbf{A}_i$ is the vector of allocation indicators. The vector-valued influence function is therefore:
        \begin{equation*}
         \boldsymbol{\psi}_{\mathbf{c}_0}(R_i) = -J_0^{-1}\mathbf{A}_i.
        \end{equation*}
        \item \textbf{Welfare Gradient $\nabla_{\mathbf{c}} \mathcal{U}(\mathbf{c}_0)$:} We apply Theorem~\ref{th:main_tech}. The derivative with respect to a single cutoff, $c_j$, is determined by the welfare change of agents at that margin ($V_{i,j}=c_j$). Let $\rho_{j \to k}$ be the density of agents with score $V_{i,j}=c_j$ who, upon losing eligibility for school $j$, are reallocated to their next-best option, $k$. Let $\tau_{j \to k}(c_j)$ be their average change in outcome. The $j$-th component of the gradient is the sum over all possible reallocations from margin $j$:
        \begin{equation*}
         \frac{\partial \mathcal{U}}{\partial c_j}(\mathbf{c}_0) = -\sum_{k \in \mathcal{A}} \rho_{j \to k} \tau_{j \to k}(c_j).
        \end{equation*}
        This corresponds to the vector $\tilde{\boldsymbol{\sigma}}$ defined in the main text.
        \item \textbf{The Jacobian Matrix $J_0$:} The entry $J_{kj} = \partial \E_0[\mu_k]/\partial c_j$ measures how enrollment at school $k$ changes when the cutoff for school $j$ increases. The diagonal elements $J_{jj}$ are negative (enrollment at $j$ falls). The off-diagonal elements $J_{kj}$ ($k\neq j$) are positive and represent the substitution effects: the density of agents who are pushed out of school $j$ and into school $k$, i.e., $J_{kj} = \rho_{j \to k}$.
    \end{itemize}
    Combining these, the market conduct externality is:
    \begin{align*}
    \Psi^{\text{conduct}}(R_i) &= \langle \nabla_{\mathbf{c}} \mathcal{U}(\mathbf{c}_0), \boldsymbol{\psi}_{\mathbf{c}_0}(R_i) \rangle = (\nabla_{\mathbf{c}} \mathcal{U}(\mathbf{c}_0))^\top (-J_0^{-1}\mathbf{A}_i) \\
    &= - \left( (J_0^{-1})^\top \nabla_{\mathbf{c}} \mathcal{U}(\mathbf{c}_0) \right)^\top \mathbf{A}_i.
    \end{align*}
    Let us define the vector of social externality values $\mathbf{v} := (J_0^{-1})^\top \nabla_{\mathbf{c}} \mathcal{U}(\mathbf{c}_0)$. Then the expression simplifies to $\Psi^{\text{conduct}}(R_i) = -\mathbf{v}^\top\mathbf{A}_i$.

    \item \textbf{Total Equilibrium-Adjusted Outcome:}
    Summing the components gives the final vector form:
    \begin{equation*}
     \Psi_i^{\text{total}} = Y_i - \mathbf{v}^\top\mathbf{A}_i = Y_i - v_1 \cdot \1\{A_i=1\} - v_2 \cdot \1\{A_i=2\}.
    \end{equation*}
    This result shows that the social value of a seat at a given school, $v_k$, is a complex combination of the marginal treatment effects at *all* cutoffs, weighted by the full matrix of equilibrium substitution patterns captured by $(J_0^{-1})^\top$.
\end{enumerate}

\subsection{Second-Price Auction with a Reserve Price}

\subsubsection{Model Specification}
We consider a sealed-bid, second-price auction for a single good ($a \in \{0,1\}$) among $n$ i.i.d. participants.
\begin{itemize}
    \item \textbf{Reports:} Each agent $i$ submits a bid $R_i \in \R_+$, which we take to be their private valuation.
    \item \textbf{Allocation Rule:} An agent wins if their bid is above the reserve price, $c$, and is the highest of all $n$ bids. The probability of winning is $\mu_1(r, c, P_R) = \ind\{r > c\} \cdot F_R(r)^{n-1}$.
    \item \textbf{Market Conduct Rule:} The reserve price $c_0$ is set to ensure the ex-ante probability of an agent winning is a fixed quantity $q$. This rule simplifies to an algebraic equation:
\begin{align*}
    \frac{1}{n}\left(1 - F_{R|0}(c_0)^n\right) = q.
\end{align*}
\end{itemize}

\subsubsection{Derivation of $\Psi_i^{\text{total}}$}
The total equilibrium-adjusted outcome is the sum of the private outcome and two distinct externality terms: $\Psi_i^{\text{total}} = Y_i + \gamma(R_i) + \Psi^{\text{conduct}}(R_i)$.

\paragraph{1. Competition Externality $\gamma(R_i)$:}
This term captures the welfare impact from a change in the bid distribution on other bidders, holding the reserve price fixed. It simplifies to an intuitive expression involving the maximum order statistic of the $n-1$ competing bids, which we denote $R_{(n-1)}$:
\begin{align*}
    \gamma(R_i) = \mathbb{E}_0\left[\tau(R_{(n-1)}) | R_{(n-1)} \ge \tilde{r}\right] \times \left(1 - F_{R|0}(\tilde{r})^{n-1}\right), \quad \text{where } \tilde{r} = \max(c_0, R_i).
\end{align*}
This is the expected treatment effect for the winning competitor, conditional on them being a relevant threat (bidding above $\tilde{r}$), multiplied by the probability that such a threat exists.

\paragraph{2. Market Conduct Externality $\Psi^{\text{conduct}}(R_i)$:}
This term captures the welfare impact from agent $i$'s influence on the equilibrium reserve price. We derive its influence function, $\psi_{c_0}(R_i)$, from the simplified market conduct rule using the Implicit Function Theorem. This yields:
\begin{align*}
    \psi_{c_0}(R_i) = \frac{\ind\{R_i > c_0\}}{f_{R|0}(c_0)}.
\end{align*}
Multiplying this by the welfare gradient, $\nabla_c \mathcal{U}(c_0) = -\tau(c_0) F_{R|0}(c_0)^{n-1} f_{R|0}(c_0)$, gives the externality:
\begin{align*}
    \Psi^{\text{conduct}}(R_i) = -\tau(c_0) F_{R|0}(c_0)^{n-1} \cdot \ind\{R_i > c_0\}.
\end{align*}
This shows the externality is non-zero only for losing bidders.

\subsection{Auction with an Optimal Reserve Price}

We now modify the previous auction example by changing the platform's objective. Instead of setting a reserve price to meet a quantity target, the platform sets it to maximize expected revenue, following Myerson (1981).

\subsubsection{Model Specification}
The setup for agents, reports, and the allocation rule is identical to the second-price auction in the previous section. The only change is the market conduct rule.
\begin{itemize}
    \item \textbf{Market Conduct Rule:} The reserve price $c$ is set to solve the platform's first-order condition for revenue maximization:
    \begin{equation*}
    G(c, P_R) := (1-F_R(c)) - c f_R(c) = 0,
    \end{equation*}
    where $f_R(c)$ is the probability density function of reports evaluated at the reserve price $c$.
\end{itemize}

\subsubsection{Verification of Assumptions}
The presence of the density term $f_R(c)$ in the market conduct rule makes its differentiability properties more subtle.
\begin{itemize}
    \item \textbf{Failure in $L_2$:} The derivative of the functional $c(P_R)$ with respect to a perturbation with score $s_R$ can be found via the Implicit Function Theorem. This derivative involves a term proportional to $f_R(c_0) s_R(c_0)$, which is a point evaluation of the score function. The point-evaluation operator is not a continuous linear functional on the space $L_2(R_i)$. A sequence of scores can converge to zero in the $L_2$ norm while diverging at the specific point $c_0$. Consequently, Assumption~\ref{ass:conduct_rule_formal} is violated for the standard $L_2$ space, and Corollary~\ref{cor:psi} does not apply.

    \item \textbf{Resolution in Sobolev Space:} To proceed, we must restrict the class of admissible policy reforms to those that induce sufficiently smooth scores. We assume the space of report scores $\mathcal{H}_R$ is a weighted Sobolev space $H^1$, with the inner product:
    \begin{equation*}
     \langle \psi, s \rangle_{H^1} := \E_0[\psi(R_i)s(R_i)] + \E_0[\psi'(R_i)s'(R_i)].
    \end{equation*}
    In this space, the Sobolev embedding theorem ensures that point evaluation is a continuous operator. Therefore, the functional $c(P_R)$ is continuously differentiable. Assumption~\ref{ass:conduct_rule_formal} now holds, but for this stronger Hilbert space. All other assumptions are maintained as before.
\end{itemize}

\subsubsection{Derivation of the MPE}
Since Corollary~\ref{cor:psi} does not apply, we must use the general form of the MPE from Theorem~\ref{th:mpe_general}:
\begin{equation*}
\text{MPE} = \E_0[\Psi_i^{\text{fixed}} s_W(W_i)] + \langle \Psi^{\text{conduct}}, s_R \rangle_{\mathcal{H}_R}.
\end{equation*}

\begin{enumerate}
    \item \textbf{Fixed Component $\Psi_i^{\text{fixed}}$:} This component is identical to the previous auction example: $\Psi_i^{\text{fixed}} = Y_i + \gamma(R_i)$.

    \item \textbf{Market Conduct Component $\langle \Psi^{\text{conduct}}, s_R \rangle_{\mathcal{H}_R}$:} This term is $\langle \nabla_c \mathcal{U}(c_0) \cdot \psi_{c_0}(R_i), s_R(R_i) \rangle_{H^1}$. We derive the representer $\psi_{c_0}$ below.
    \begin{itemize}
        \item \textbf{Welfare Gradient $\nabla_c \mathcal{U}(c_0)$:} This is identical to the previous auction example:
        \begin{equation*}
         \nabla_c \mathcal{U}(c_0) = -\tau(c_0) f_{R|0}(c_0) F_{R|0}(c_0)^{n-1}.
        \end{equation*}
        \item \textbf{Characterization of the Representer $\psi_{c_0}$:} The representer is defined by the relation $c'[s_R] = \langle \psi_{c_0}, s_R \rangle_{H^1}$. We first find an expression for the functional $c'[s_R]$ using the Implicit Function Theorem on $G(c, P_R)=0$. Differentiating w.r.t. the policy perturbation at baseline gives:
        \begin{equation*}
         \frac{\partial G}{\partial c}\bigg|_{c_0, P_{R|0}} \cdot c'[s_R] + \frac{\partial G}{\partial P_R}[s_R]\bigg|_{c_0, P_{R|0}} = 0.
        \end{equation*}
        The partial derivatives are $\frac{\partial G}{\partial c} = -2f_R(c) - c f'_R(c)$ and
        $\frac{\partial G}{\partial P_R}[s_R] = -\E_0[s_R(R_i)\1\{R_i \le c_0\}] - c_0 f_{R|0}(c_0)s_R(c_0)$.
        Solving for $c'[s_R]$ gives the functional:
        \begin{equation*}
         c'[s_R] = K \cdot \left( \E_0[s_R(R_i)\1\{R_i \le c_0\}] + c_0 f_{R|0}(c_0)s_R(c_0) \right),
        \end{equation*}
        where $K = (2f_{R|0}(c_0) + c_0 f'_{R|0}(c_0))^{-1}$. We now equate this with the inner product form:
        \begin{equation*}
         \E_0[\psi_{c_0} s_R] + \E_0[\psi'_{c_0} s'_R] = K \cdot \left( \E_0[s_R \cdot \1_{\le c_0}] + c_0 f_{R|0}(c_0)s_R(c_0) \right).
        \end{equation*}
        Using integration by parts on the second term on the left, $\E_0[\psi'_{c_0} s'_R] = -\E_0[s_R (\psi'_{c_0}f_{R|0})'/f_{R|0}]$, and representing the point evaluation using the Dirac delta function, $s_R(c_0) = \E_0[s_R(R_i) \delta(R_i-c_0)/f_{R|0}(c_0)]$, we can group all terms under a single expectation over $s_R(R_i)$:
        \begin{equation*}
         \E_0\left[ s_R(R_i) \left( \psi_{c_0}(R_i) - \frac{(\psi'_{c_0}(R_i)f_{R|0}(R_i))'}{f_{R|0}(R_i)} - K \cdot (\1\{R_i \le c_0\} + c_0\delta(R_i-c_0)) \right) \right] = 0.
        \end{equation*}
        Since this must hold for all $s_R \in H^1$, the term in the parentheses must be zero. This yields the Sturm-Liouville differential equation that uniquely defines the representer $\psi_{c_0}(r)$:
        \begin{equation*}
         \psi_{c_0}(r)f_{R|0}(r) - (\psi'_{c_0}(r)f_{R|0}(r))' = K \cdot \left( \1\{r \le c_0\}f_{R|0}(r) + c_0 f_{R|0}(c_0)\delta(r-c_0) \right).
        \end{equation*}
    \end{itemize}
    \item \textbf{Total Marginal Policy Effect:}
    The MPE is therefore:
    \begin{align*}
    \text{MPE} = \E_0[(Y_i + \gamma(R_i))s_W(W_i)] + \nabla_c \mathcal{U}(c_0) \cdot \langle \psi_{c_0}, s_R \rangle_{H^1}.
    \end{align*}
    This expression cannot be simplified into a single covariance. It demonstrates that when the market-clearing rule is itself the solution to an optimization problem that depends on local features of the report distribution, the welfare impact of a policy reform depends not just on its direction ($s_R$), but also on its smoothness (via the derivative term $s'_R$ implicit in the $H^1$ inner product).
\end{enumerate}

\subsection{Top Trading Cycles}

This section discusses the Top Trading Cycles (TTC) mechanism. Unlike the previous examples, the market-clearing parameters (the cutoffs) are defined as the solution to a dynamic system. We first describe this characterization and then provide a high-level sketch of our analysis.

\subsubsection{Model Specification and Cutoff Characterization}
We follow the continuum model of \citet{leshno2021cutoff}. The market consists of a continuum of students and a finite set of schools $\mathcal{C}=\{1,...,n\}$ with capacities $q_c$. A student's type $R_i$ includes their strict preference ordering $\succ_i$ and a vector of priority scores $V_i = (V_{i,1}, ..., V_{i,n}) \in [0,1]^n$.

\begin{itemize}
    \item \textbf{Allocation Rule and Cutoffs:} The TTC algorithm assigns students by clearing trading cycles. \citet{leshno2021cutoff} show that the final assignment can be described by a matrix of cutoffs $\mathbf{c} = \{c_{a,b}\}_{a,b \in \mathcal{C}}$. A student $i$ is admitted to their most-preferred school $a$ within their "budget set," which is the set of schools $B(R_i, \mathbf{c}) = \{k \in \mathcal{C} \mid \exists b \in \mathcal{C} \text{ s.t. } V_{i,b} \ge c_{k,b}\}$

    \item \textbf{Market Conduct Rule (Dynamic System):} The cutoff matrix $\mathbf{c}$ is not determined by a simple set of algebraic equations but as the solution to a dynamic process. The key objects are:
    \begin{enumerate}
        \item \textbf{The TTC Path $\boldsymbol{\gamma}(t)$:} A vector-valued function $\boldsymbol{\gamma}(t) \in [0,1]^n$ that tracks the priority frontiers of the schools over time, starting from $\boldsymbol{\gamma}(0)=\mathbf{1}$.
        \item \textbf{Marginal Trade Balance Equations:} The path evolves according to a system of ordinary differential equations that ensure the "flow" of students trading into a school equals the "flow" of students trading out of it at every moment. For each school $k$, this is:
        \begin{equation*}
         \sum_{b \in \mathcal{C}} \gamma'_b(t) H_b^k(\boldsymbol{\gamma}(t)) = \sum_{a \in \mathcal{C}} \gamma'_k(t) H_k^a(\boldsymbol{\gamma}(t)).
        \end{equation*}
       Here, $H_b^k(\mathbf{x})$ is the marginal density of students at the priority frontier $\mathbf{x}$ whose top choice is school $k$ and who have the highest priority at school $b$.
        \item \textbf{Capacity Equations (Stopping Conditions):} The process for a school $k$ stops at time $t^{(k)}$ when its capacity is filled. The final cutoffs are determined by the path evaluated at these stopping times: $c_{k,b} = \gamma_b(t^{(k)})$.
    \end{enumerate}
\end{itemize}
The market conduct rule $\mathbf{c}(P_R)$ is therefore the function that maps a distribution of reports $P_R$ (which determines the marginal densities $H_b^k$) to the matrix of cutoffs that solves this dynamic system.

\subsubsection{TTC in the MPE Framework}
For a given cutoff matrix $\mathbf{c}$, the allocation rule $\mu_a(R_i, \mathbf{c})$ is a complex but deterministic indicator function. It can be written in our decomposition form with $h_a=1$ and $\phi_a(R_i, \mathbf{c})$ representing the condition that $a$ is the most-preferred school in the budget set $B(R_i, \mathbf{c})$.
Since $h_a=1$, the allocation rule does not depend on $P_R$ once $\mathbf{c}$ is known. Therefore, the Hadamard derivative kernel $L_a$ is zero, and the competition externality $\gamma(R_i)$ is zero. The entire equilibrium spillover is captured by the market conduct effect. We assume the primitives are regular enough for the assumptions of Corollary~\ref{cor:psi} to hold.

\subsubsection{Derivation for a Parametric Case}
Let's analyze an economy with two schools ($n=2$), capacities $q_1, q_2$, and a unit mass of students. A fraction $\pi_1$ of students prefer school 1, and the remaining $\pi_2=1-\pi_1$ prefer school 2. Priorities $V_i=(V_{i,1}, V_{i,2})$ are independently and uniformly distributed on $[0,1]^2$ and are independent of preferences.

\begin{enumerate}
    \item \textbf{Solving the Dynamic System:}
    In this setting, the marginal densities are constant: $H_b^k(\mathbf{x}) = \pi_k$ for any $b,k \in \{1,2\}$. The trade balance equation for school 1 becomes:
    \begin{equation*}
     \gamma'_1(t) H_1^1 + \gamma'_2(t) H_2^1 = \gamma'_1(t) H_1^1 + \gamma'_1(t) H_1^2 \implies \gamma'_2(t) \pi_1 = \gamma'_1(t) \pi_2.
    \end{equation*}
    This gives the simple linear ODE $\gamma'_2(t)/\gamma'_1(t) = \pi_2/\pi_1$. Parameterizing the path by $t$ such that $\gamma_1(t) = 1-t$, the solution with initial condition $\boldsymbol{\gamma}(0)=(1,1)$ is the line:
    \begin{equation*}
     \gamma_2(t) = 1 - (\pi_2/\pi_1)t.
    \end{equation*}
    \item \textbf{Characterizing the Welfare Gradient $\nabla_{\mathbf{c}} \mathcal{U}$:}
    The parameter vector $\mathbf{c}$ is the matrix of four cutoffs $\{c_{1,1}, c_{1,2}, c_{2,1}, c_{2,2}\}$. Applying Theorem~\ref{th:main_tech}, the gradient $\partial \mathcal{U}/\partial c_{k,b}$ is identified by the welfare change of agents on the boundary $V_{i,b} = c_{k,b}$ who are reallocated. We denote this gradient abstractly by $\nabla_{\mathbf{c}} \mathcal{U}$.

    \item \textbf{Characterizing the Influence Function $\boldsymbol{\psi}_{\mathbf{c}_0}(R_i)$:}
    The influence function is the vector of derivatives of the cutoffs with respect to a perturbation in the distribution. A general policy perturbation with score $s_R$ on the report distribution $P_R$ induces a pathwise derivative on any functional of that distribution. We find the influence function by differentiating the entire dynamic system that determines the cutoffs.

    \begin{itemize}
        \item \textbf{Perturbation of Primitives:} The perturbation directly affects the marginal densities $H_b^k(\boldsymbol{\gamma})$ that govern the path's evolution, and the demand functions $D^k(\boldsymbol{\gamma})$ that determine the stopping times. We denote their pathwise derivatives as $H_b^{k\prime}[s_R]$ and $D^{k\prime}[s_R]$. For example, in our parametric case, the perturbation affects both the shares $\pi_k$ and the uniformity of the priority distribution.

        \item \textbf{Path Influence Function:} The TTC path is defined by the ODE $\gamma'_2(t) H_2^1(\boldsymbol{\gamma}(t)) = \gamma'_1(t) H_1^2(\boldsymbol{\gamma}(t))$. We differentiate this entire equation with respect to the policy perturbation. This yields a variational equation for the influence on the path, $\boldsymbol{\psi}_\gamma(t; R_i)$, which now includes a forcing term due to the direct perturbation of the $H$ functions:
        \begin{equation*}
         \frac{d}{dt}\psi_{\gamma,2}(t) = J(\boldsymbol{\gamma}(t)) \cdot \psi_{\gamma,2}(t) + F(\boldsymbol{\gamma}(t), s_R).
        \end{equation*}
        Here, $J$ is a Jacobian term from differentiating the ODE's coefficients with respect to $\boldsymbol{\gamma}$, and the forcing term $F$ is a linear functional of the score $s_R$ that depends on the derivatives $H_b^{k\prime}[s_R]$. The solution to this ODE gives the influence function for the path shape.

        \item \textbf{Stopping Time Influence Function:} Assume school 1 fills first. The stopping time $t^{(1)}$ is implicitly defined by the capacity constraint $D^1(\boldsymbol{\gamma}(t^{(1)}), P_R) = q_1$. Differentiating this constraint with respect to the perturbation at baseline gives:
        \begin{equation*}
         \frac{\partial D^1}{\partial P_R}[s_R] + \nabla_{\boldsymbol{\gamma}}D^1 \cdot \left( \boldsymbol{\gamma}'(t^{(1)}_0)\frac{d t^{(1)}}{d\theta}[s_R] + \frac{d\boldsymbol{\gamma}(t^{(1)}_0)}{d\theta}[s_R] \right) = 0.
        \end{equation*}
        The term $\frac{\partial D^1}{\partial P_R}[s_R]$ captures the direct effect of the perturbation on the demand functional. Solving for the derivative of the stopping time, $\frac{d t^{(1)}}{d\theta}[s_R]$, yields its influence function $\psi_{t^{(1)}}(R_i)$:
        \begin{equation*}
         \psi_{t^{(1)}}(R_i) = -\left(\frac{d D^1}{dt}\bigg|_{t^{(1)}_0}\right)^{-1} \left( D^{1\prime}[R_i] + \nabla_{\boldsymbol{\gamma}}D^1 \cdot \boldsymbol{\psi}_\gamma(t^{(1)}_0; R_i) \right),
        \end{equation*}
        where $D^{1\prime}[R_i]$ is the representer for the pathwise derivative of the demand functional.

        \item \textbf{Cutoff Influence Functions:} The influence functions for the cutoffs are found by applying the chain rule. The derivation for the second-round cutoffs simplifies considerably. As established by \citet{leshno2021cutoff}, once school 1 fills at time $t^{(1)}$, its priority frontier stops advancing, i.e., $\gamma_1(t) = \gamma_1(t^{(1)})$ for all $t \ge t^{(1)}$.
            \begin{itemize}
                \item \textbf{The cutoff $c_{1,2}$:} This cutoff is defined as $\gamma_1(t^{(2)})$. Since $t^{(2)} \ge t^{(1)}$, it follows that $\gamma_1(t^{(2)}) = \gamma_1(t^{(1)}) = c_{1,1}$. Thus, we have the identity $c_{1,2} = c_{1,1}$, which implies their influence functions must also be equal: $\psi_{c_{1,2}}(R_i) = \psi_{c_{1,1}}(R_i)$.
                \item \textbf{The cutoff $c_{2,2}$:} The final cutoff, $c_{2,2} = \gamma_2(t^{(2)})$, is found by solving the capacity constraint for school 2 in the residual economy. Its influence function is found by differentiating this expression: $\psi_{c_{2,2}}(R_i) = \gamma'_{2}(t^{(2)}_0)\psi_{t^{(2)}}(R_i) + \psi_{\gamma,2}(t^{(2)}_0; R_i)$. The term $\psi_{t^{(2)}}(R_i)$ is the influence function for the second stopping time, which is found by differentiating the capacity constraint for the residual economy. This constraint depends on the outcomes of the first round, making $\psi_{t^{(2)}}(R_i)$ a function of the previously derived first-round influence functions.
            \end{itemize}
    \end{itemize}
This more general procedure fully identifies the vector of influence functions, $\boldsymbol{\psi}_{\mathbf{c}_0}(R_i)$, for any perturbation to the report distribution that satisfies our regularity conditions.

    \item \textbf{Total Equilibrium-Adjusted Outcome:}
    The final expression is:
    \begin{equation*}
     \Psi_i^{\text{total}} = Y_i - \langle \nabla_{\mathbf{c}} \mathcal{U}(\mathbf{c}_0), \boldsymbol{\psi}_{\mathbf{c}_0}(R_i) \rangle.
    \end{equation*}
    This derivation confirms that even for a complex procedural mechanism like TTC, the MPE can be constructed within our framework. The market conduct externality is fully characterized by the welfare gradient at the cutoff boundaries and the influence function of the cutoffs, which is itself found by analyzing the sensitivity of the underlying dynamic system to policy perturbations.
\end{enumerate}

\newpage

\section{Extensions}\label{app:extensions}

This appendix provides formal derivations for the extensions discussed in Section~\ref{sec:extensions} of the main text.

\subsection{General Welfare Functionals}
This section demonstrates that our framework extends from the mean to a general class of welfare criteria. We first provide a simple proof for functionals defined by moment conditions using the Implicit Function Theorem. We then provide a more abstract, general proof that covers any Hadamard differentiable functional.

\subsubsection{A Direct Proof for Functionals Defined by Moment Conditions}
Many common statistics, like quantiles, are defined implicitly as the solution to a moment equation. For this large class of functionals, we can prove the main result directly using the Implicit Function Theorem.

\begin{proposition}
Let the welfare functional $\mathcal{U}$ be defined implicitly as the unique solution to a moment equation $\E[g(Y_i, \mathcal{U})] = 0$. Assume the function $g(y,u)$ is continuously differentiable in $u$ and that $\E_0[\partial g(Y_i, \mathcal{U}_0)/\partial \mathcal{U}] \neq 0$. The Marginal Policy Effect on $\mathcal{U}$ is given by applying Theorem~\ref{th:mpe_general} (or Corollary~\ref{cor:psi}) to the transformed outcome $Z_i = IF(Y_i; \mathcal{U}_0)$, where $IF(y; u) = -\left(\E[\frac{\partial g}{\partial u}]\right)^{-1} g(y, u)$ is the influence function of the functional $\mathcal{U}$.
\end{proposition}
\begin{proof}
Under a policy perturbation $\theta$, the moment condition must hold for the perturbed value of the functional, $\mathcal{U}(\theta)$:
\begin{equation*}
G(\theta, \mathcal{U}(\theta)) := \E_\theta[g(Y_i, \mathcal{U}(\theta))] = 0.
\end{equation*}
Our goal is to find the MPE, $\frac{d\mathcal{U}}{d\theta}|_{\theta=0}$. By the Implicit Function Theorem:
\begin{equation*}
\frac{d\mathcal{U}}{d\theta}\bigg|_{\theta=0} = - \left( \frac{\partial G}{\partial \mathcal{U}}\bigg|_{0, \mathcal{U}_0} \right)^{-1} \left( \frac{\partial G}{\partial \theta}\bigg|_{0, \mathcal{U}_0} \right).
\end{equation*}
The first component is $\frac{\partial G}{\partial \mathcal{U}} = \E_0[\frac{\partial g(Y_i, \mathcal{U}_0)}{\partial \mathcal{U}}]$. The second component is the MPE for the outcome variable $Z_i^g := g(Y_i, \mathcal{U}_0)$:
\begin{equation*}
 \frac{\partial G}{\partial \theta} = \frac{d}{d\theta}\E_\theta[g(Y_i, \mathcal{U}_0)]\bigg|_{\theta=0} = \mathcal{L}(g),
\end{equation*}
where $\mathcal{L}(g)$ is the MPE operator. From Appendix~\ref{app:mpe_proof}, we know $\mathcal{L}(g) = \E_0[\Psi_i^g s_W(W_i)]$ (in the $L_2$ case).
Substituting these into the IFT formula gives:
\begin{equation*}
    \frac{d\mathcal{U}}{d\theta} = - \left(\E_0\left[\frac{\partial g}{\partial \mathcal{U}}\right]\right)^{-1} \E_0[\Psi_i^g s_W(W_i)] = \E_0\left[ - \left(\E_0\left[\frac{\partial g}{\partial \mathcal{U}}\right]\right)^{-1} \Psi_i^g \cdot s_W(W_i) \right].
\end{equation*}
The term inside the expectation is the equilibrium-adjusted outcome for $\mathcal{U}$. The influence function for $\mathcal{U}$ is $IF(Y; \mathcal{U}_0) = -(\E_0[\frac{\partial g}{\partial \mathcal{U}}])^{-1} g(Y, \mathcal{U}_0)$. Since the construction of the externality terms in $\Psi$ is linear, it follows that $-(\E_0[\frac{\partial g}{\partial \mathcal{U}}])^{-1}\Psi_i^g = \Psi_i^{IF}$. This shows that the MPE for $\mathcal{U}$ is $\E_0[\Psi_i^{IF} s_W(W_i)]$, which completes the proof.
\end{proof}

\subsubsection{The General Case for Hadamard Differentiable Functionals}
The result holds more generally for any Hadamard differentiable functional. The proof below uses an integration-by-parts argument that relies on the linearity of the MPE operator and an assumption of bounded support for the outcome variable.

\begin{lemma}[Continuity of the MPE Operator]
The MPE operator $\mathcal{L}(g) := \frac{d}{d\theta} \E_\theta[g(Y_i)] |_{\theta=0}$ is a continuous (bounded) linear operator on the space of bounded, measurable functions $g$, equipped with the sup norm $||g||_\infty = \sup_y |g(y)|$.
\end{lemma}

\begin{proof}
Linearity follows from the linearity of expectations and the construction of $\Psi^g$. For continuity, we show the operator is bounded. From Corollary~\ref{cor:psi}, $\mathcal{L}(g) = \E_0[\Psi^g s_W]$. By Cauchy-Schwarz, $|\mathcal{L}(g)| \le \sqrt{\E_0[(\Psi^g)^2]} \cdot \sqrt{\E_0[s_W^2]}$. The externality terms in $\Psi^g$ are constructed via bounded linear operations on the conditional mean of $g$, which is itself bounded by $||g||_\infty$. It follows that $\sqrt{\E_0[(\Psi^g)^2]}$ is bounded by a constant times $||g||_\infty$. Thus, $|\mathcal{L}(g)| \le C \cdot ||g||_\infty$ for some constant $C$.
\end{proof}

\begin{proposition}
Let $\mathcal{U}(F_Y)$ be a Hadamard differentiable functional with a continuously differentiable influence function $IF(y; F_{Y|0})$. Assume the outcome variable $Y$ has bounded support. The Marginal Policy Effect on $\mathcal{U}$ is given by applying Theorem~\ref{th:mpe_general} (or Corollary~\ref{cor:psi}) to the transformed outcome $IF(Y_i; F_{Y|0})$.
\end{proposition}

\begin{proof}
The proof establishes the identity $\text{MPE}(\mathcal{U}) = \text{MPE}(\E[IF(Y)])$.

\textbf{Step 1: MPE($\mathcal{U}$) as a Stieltjes Integral.}
By definition of the Hadamard derivative, the MPE of $\mathcal{U}$ is the integral of its influence function against the derivative of the path of outcome measures, $\nu'$. Let the bounded support of $Y$ be $[a,b]$.
\begin{equation*}
    \text{MPE}(\mathcal{U}) = \int_{a}^{b} IF(y; F_{Y|0}) \, d\nu'(y), \quad \text{where} \quad \nu'([a, y]) = \mathcal{L}(\1\{Y \le y\}).
\end{equation*}

\textbf{Step 2: Integration by Parts.}
We apply the integration by parts formula for Stieltjes integrals. Let $u(y) = IF(y)$ and $dv = d\nu'(y)$, which implies $v(y) = \mathcal{L}(\1\{Y \le y\})$.
\begin{equation*}
    \text{MPE}(\mathcal{U}) = \left[ IF(y) \cdot \mathcal{L}(\1\{Y \le y\}) \right]_{y=a}^{b} - \int_{a}^{b} \mathcal{L}(\1\{Y \le y\}) \cdot IF'(y) dy.
\end{equation*}
The boundary terms are zero. At $y=b$, $\mathcal{L}(\1\{Y \le b\}) = \mathcal{L}(1) = 0$. At $y=a$, $\mathcal{L}(\1\{Y \le a\}) = \mathcal{L}(0) = 0$.

\textbf{Step 3: Swapping Linear Operators.}
We are left with the integral term. By the preceding lemma, $\mathcal{L}$ is a continuous linear operator and thus commutes with the integral:
\begin{equation*}
    \text{MPE}(\mathcal{U}) = - \int_{a}^{b} \mathcal{L}(\1\{Y \le y\}) \cdot IF'(y) dy = - \mathcal{L}\left( \int_{a}^{b} IF'(y) \cdot \1\{Y \le y\} dy \right).
\end{equation*}

\textbf{Step 4: Conclusion.}
The inner integral, for a fixed realization of $Y$, is $\int_Y^{b} IF'(y) dy = [IF(y)]_Y^b = IF(b) - IF(Y)$.
Substituting this back, the MPE is:
\begin{equation*}
    \text{MPE}(\mathcal{U}) = - \mathcal{L}\left( IF(b) - IF(Y) \right) = - \mathcal{L}(IF(b)) + \mathcal{L}(IF(Y)).
\end{equation*}
Since $IF(b)$ is a constant, its MPE is zero, so $\mathcal{L}(IF(b)) = 0$. This leaves the final identity:
\begin{equation*}
    \text{MPE}(\mathcal{U}) = \mathcal{L}(IF(Y)) = \text{MPE}(\E_0[IF(Y)]).
\end{equation*}
This shows that computing the MPE for a general functional is equivalent to computing the MPE for the mean of its influence function, which our main framework is designed to do.
\end{proof}

\end{document}